%% file: ms.tex
\documentclass[10pt,journal,twocolumn,twoside]{IEEEtran}

\usepackage{amsfonts,amsbsy,bm,amsmath}
\usepackage[nolist]{acronym}
\usepackage{mathrsfs} 
\usepackage{graphicx,cite,amssymb,mathtools,amsthm}
\usepackage{tikz,adjustbox}
\usepackage{stfloats}
\usepackage{subcaption}
\usepackage{pgfplots}
\usepackage{scalerel}
\usepackage{nicefrac}
\usepackage{stmaryrd}

\pgfplotsset{compat=newest,compat/show suggested version=false}
\mathtoolsset{showonlyrefs}
\usepackage{bbm}
\allowdisplaybreaks
\DeclareCaptionLabelSeparator{periodspace}{.\quad}
\newcommand{\mcc}{\textcolor{black}}
\newcommand{\mc}{\textcolor{black}}
\newcommand{\last}{\textcolor{black}}

\newcommand{\lasttt}{\textcolor{black}}

\newcommand{\own}{\textcolor{black}} %activate in the camera ready to see the changes after revision
\newcommand{\ed}{\textcolor{black}}

\newcommand{\GL}{\textcolor{black}}

\DeclareMathSizes{10}{10}{5}{4}
\IEEEoverridecommandlockouts
\begin{document}
	\input{notation}

	\title{Successive Cancellation Decoding of\\Single Parity-Check Product Codes:\\Analysis and Improved Decoding}
	\author{Mustafa Cemil Co\c{s}kun, \IEEEmembership{Student Member, IEEE}, Gianluigi Liva, \IEEEmembership{Senior Member, IEEE},\\ Alexandre Graell i Amat, \IEEEmembership{Senior Member, IEEE}, Michael Lentmaier, \IEEEmembership{Senior Member, IEEE},\\ and Henry D. Pfister, \IEEEmembership{Senior Member, IEEE}
	\thanks{This work was supported by the Helmholtz Gemeinschaft through the HGF-Allianz DLR@Uni project Munich Aerospace via the research grant ``Efficient Coding and Modulation for Satellite Links with Severe Delay Constraints''. This paper was presented in part at the IEEE International Symposium on Information Theory, June 2017, Aachen, Germany \cite{coskun17}.}
	\thanks{Mustafa Cemil Co\c{s}kun is with the Institute for Communications Engineering (LNT), Technical University of Munich (TUM), Munich, Germany (email: mustafa.coskun@tum.de). This work was carried out when he was also with the Institute of Communications and Navigation of the German Aerospace Center (DLR), Weßling, Germany.}
	\thanks{Gianluigi Liva is with the Institute of Communications and Navigation of the DLR, We{\ss}ling, Germany (email: gianluigi.liva@dlr.de).}
	\thanks{Alexandre Graell i Amat is with the Department of Signals and Systems, Chalmers University of Technology, Gothenburg, Sweden (email: alexandre.graell@chalmers.se).}
	\thanks{Michael Lentmaier is with the Department of Electrical and Information Technology, Lund University, Lund, Sweden (email: michael.lentmaier@eit.lth.se).}
	\thanks{Henry D. Pfister is with the Department of Electrical and Computer Engineering, Duke University, Durham, USA (email: henry.pfister@duke.edu).}
	} 

% make the title area
	\maketitle

% As a general rule, do not put math, special symbols or citations
% in the abstract
	%\vspace{-15mm}
	\begin{abstract}
		A product code with single parity-check component codes can be \ed{described via the tools of} a multi-kernel polar code, where the rows of the generator matrix are chosen according to the constraints imposed by the product code construction. Following this observation, successive cancellation decoding of such codes is introduced. In particular, the error probability of single parity-check product codes over \ed{binary memoryless symmetric channels} under successive cancellation decoding is characterized. A bridge with the analysis of product codes introduced by Elias is also established \ed{for the binary erasure channel}. Successive cancellation list decoding of single parity-check product codes is then described. For the provided example, \ed{simulations over the binary input additive white Gaussian channel show} that successive cancellation list decoding outperforms belief propagation decoding \ed{applied to} the code graph. Finally, the performance of the concatenation of a product code with a high-rate outer code is investigated via distance spectrum analysis. \ed{Examples of concatenations performing within $0.7$ dB from the random coding union bound are provided.}
	\end{abstract}
	\markboth
	{submitted to IEEE Transactions on Information Theory}
	{}

	\input{introduction.tex}
	\input{preliminaries.tex}
	\input{spc_pc_as_polar_code.tex}
	\input{sc.tex}
	\input{analysis.tex}

	\input{sc_list.tex}
	\input{Conclusions.tex}

	\section*{Acknowledgement}
	The authors would like to thank the associate editor and the anonymous reviewers for their valuable comments, which improved the presentation of the work significantly.
	\input{appendix.tex}
	%\IEEEtriggeratref{34}
	%\bibliographystyle{IEEEtran}
	\bibliography{IEEEabrv,product}

\end{document}

%% file: notation.tex
\newcommand{\ensemble}{\mathscr{C}}
\newcommand{\code}{\mathcal{C}}
\newcommand{\vecu}{\boldsymbol{u}}
\newcommand{\vecd}{\boldsymbol{d}}
\newcommand{\vecv}{\boldsymbol{v}}
\newcommand{\vecw}{\boldsymbol{w}}
\newcommand{\vecp}{\boldsymbol{p}}
\newcommand{\vecpsi}{\boldsymbol{\psi}}
\newcommand{\vecx}{\boldsymbol{x}}
\newcommand{\vecone}{\boldsymbol{1}}
\newcommand{\vecuhat}{\hat{\boldsymbol{u}}}
\newcommand{\vecc}{\boldsymbol{c}}
\newcommand{\vecb}{\boldsymbol{b}}
\newcommand{\vecchat}{\hat{\boldsymbol{c}}}
\newcommand{\vecy}{\boldsymbol{y}}
\newcommand{\vecz}{\boldsymbol{z}}
\newcommand{\G}{\boldsymbol{G}}
\newcommand{\GK}{\boldsymbol{\mathsf{K}}}
\newcommand{\Per}{\boldsymbol{\Pi}}
\newcommand{\I}{\boldsymbol{I}}
\newcommand{\perP}{\boldsymbol{P}}
\newcommand{\perS}{\boldsymbol{S}}
\newcommand{\GH}[1]{\bm{\mathsf{G}}_{#1}}
\newcommand{\T}[1]{\bm{\mathsf{T}}{(#1)}}

\newcommand{\mi}{\mathrm{I}}
\newcommand{\entropy}{\mathrm{H}}
\newcommand{\Prob}{\boldsymbol{P}}
\newcommand{\Z}{\mathrm{Z}}
\newcommand{\SPC}{\mathcal{S}}
\newcommand{\Rep}{\mathcal{R}}
\newcommand{\f}{\mathrm{f}}
\newcommand{\W}{\mathrm{W}}
\newcommand{\A}{\mathrm{A}}
\newcommand{\SC}{\mathrm{SC}}
\newcommand{\Q}{\mathrm{Q}}
\newcommand{\MAP}{\mathrm{MAP}}
\newcommand{\E}{\mathrm{E}}
\newcommand{\bin}{\boldsymbol{b}}
\newcommand{\U}{\mathrm{U}}
\newcommand{\Ham}{\mathrm{H}}

\newcommand{\Amin}{\mathrm{A}_{\mathrm{min}}}
\newcommand{\dmin}{d}
\newcommand{\erfc}{\mathrm{erfc}}

\newcommand{\unaryminus}{\scalebox{0.5}[1.0]{\( - \)}}
\newcommand{\unaryplus}{\scalebox{0.5}[0.5]{\( + \)}}

\newcommand{\de}{\mathrm{d}}

\newcommand{\decoRule}{\rule{\textwidth}{.4pt}}

\newcommand{\oleq}[1]{\overset{\text{(#1)}}{\leq}}
\newcommand{\oeq}[1]{\overset{\text{(#1)}}{=}}
\newcommand{\ogeq}[1]{\overset{\text{(#1)}}{\geq}}
\newcommand{\ogeql}[2]{\overset{#1}{\underset{#2}{\gtreqless}}}
\newcommand{\inner}[1]{\langle#1\rangle}

\newcommand{\pscd}{\gl{\Prob(\mathcal{E}_{\SCD})}}
\newcommand{\uED}{\gl{\hat{u}}}
\newcommand{\LED}{\gl{L_i^{(\ED)}}}
\newcommand{\inverson}[1]{\gl{\mathbb{I}\left\{#1\right\}}}

\newtheorem{mydef}{Definition}
\newtheorem{prop}{Proposition}
\newtheorem{theorem}{Theorem}

\newtheorem{lemma}{Lemma}
\newtheorem{remark}{Remark}
\newtheorem{example}{Example}
\newtheorem{definition}{Definition}
\newtheorem{corollary}{Corollary}

\renewcommand{\qed}{\hfill\ensuremath{\blacksquare}}

%%%%%%%%%%%%%%%%%%% Colors %%%%%%%%%%%%%%%%%%%%%%%%%%%%%%%%
\definecolor{lightblue}{rgb}{0,.5,1}
\definecolor{normemph}{rgb}{0,.2,0.6}
\definecolor{supremph}{rgb}{0.6,.2,0.1}
\definecolor{lightpurple}{rgb}{.6,.4,1}
\definecolor{gold}{rgb}{.6,.5,0}
\definecolor{orange}{rgb}{1,0.4,0}
\definecolor{hotpink}{rgb}{1,0,0.5}
\definecolor{newcolor2}{rgb}{.5,.3,.5}
\definecolor{newcolor}{rgb}{0,.3,1}
\definecolor{newcolor3}{rgb}{1,0,.35}
\definecolor{darkgreen1}{rgb}{0, .35, 0}
\definecolor{darkgreen}{rgb}{0, .6, 0}
\definecolor{darkred}{rgb}{.75,0,0}
\definecolor{midgray}{rgb}{.8,0.8,0.8}
\definecolor{darkblue}{rgb}{0,.25,0.6}

\definecolor{lightred}{rgb}{1,0.9,0.9}
\definecolor{lightblue}{rgb}{0.9,0,0.0}
\definecolor{lightpurple}{rgb}{.6,.4,1}
\definecolor{gold}{rgb}{.6,.5,0}
\definecolor{orange}{rgb}{1,0.4,0}
\definecolor{hotpink}{rgb}{1,0,0.5}
\definecolor{darkgreen}{rgb}{0, .6, 0}
\definecolor{darkred}{rgb}{.75,0,0}
\definecolor{darkblue}{rgb}{0,0,0.6}

\definecolor{bgblue}{RGB}{245,243,253}
\definecolor{ttblue}{RGB}{91,194,224}

\definecolor{dark_red}{RGB}{150,0,0}
\definecolor{dark_green}{RGB}{0,150,0}
\definecolor{dark_blue}{RGB}{0,0,150}
\definecolor{dark_pink}{RGB}{80,120,90}

\begin{acronym}
	\acro{B-AWGN}{binary input additive white Gaussian noise}
	\acro{B-DMC}{binary-input discrete memoryless channel}
	\acro{B-DMSC}{binary-input discrete memoryless symmetric channel}
	\acro{BMSC}{binary-input memoryless symmetric channel}
	\acro{BCJR}{Bahl, Cocke, Jelinek, and Raviv}
	\acro{BEC}{binary erasure channel}
	\acro{BER}{bit error rate}
	\acro{BLER}{block error rate}
	\acro{BMS}{binary memoryless symmetric}
	\acro{BP}{belief propagation}
	\acro{BPSK}{binary phase shift keying}
	\acro{BSC}{binary symmetric channel}
	\acro{CER}{codeword error rate}
	\acro{CN}{check node}
	\acro{CRC}{cyclic redundancy check}
%	\acro{DE}{density evolution}
	\acro{G-JIR-WEF}{generalized joint input-redundancy weight enumerator function}
	\acro{G-JIO-WEF}{generalized joint input-output weight enumerator function}
	\acro{G-JWEF}{generalized joint weight enumerator function}
%	\acro{GJWE}{generalized joint WE}
	\acro{IO-WE}{input-output weight enumerator}
%	\acro{IR-WE}{input-redundancy weight enumerator}
	\acro{IOWEF}{input-output weight enumerator function}
	\acro{IRWEF}{input-redundancy weight enumerator function}
%	\acro{JIO-WE}{joint IO-WE}
	\acro{JIOWEF}{joint input-output weight enumerator function}
%	\acro{JIR-WE}{joint IR-WE}
	\acro{JWEF}{joint weight enumerator function}
	\acro{LDPC}{low-density parity-check}
	\acro{LHS}{left-hand side}
	\acro{LLR}{log-likelihood ratio}
	\acro{MAP}{maximum a posteriori}
	\acro{MC}{metaconverse}
	\acro{ML}{maximum-likelihood}
	%\acro{PC}{Product code}
	\acro{p.d.f.}{probability density function}
	\acro{RCB}{random coding bound}
	\acro{RCU}{random coding union}
	\acro{RM}{Reed-Muller}
	\acro{RHS}{right-hand side}
	\acro{RV}{random variable}
	\acro{SPC}{single parity-check}
	\acro{SC}{successive cancellation}
	\acro{SCC}{super component codes}
	\acro{SCL}{successive cancellation list}
	\acro{SISO}{soft-input soft-output}
	\acro{SNR}{signal-to-noise ratio}
	\acro{TSB}{tangential-sphere bound}
	\acro{UB}{union bound}
	\acro{VN}{variable node}
%	\acro{WE}{weight enumerator}
	\acro{WEF}{weight enumerator function}
\end{acronym}

%% file: introduction.tex
\section{Introduction}\label{sec:intro}

\IEEEPARstart{P}{roduct} codes were introduced in 1954 by Elias\cite{Elias:errorfreecoding54} \mcc{with} extended Hamming \mcc{component} codes over an infinite number of dimensions. \mcc{Elias} showed that \last{this code has positive rate and its} bit error probability can be made arbitrarily small over the \ac{BSC}. \last{His decoder treats} the product code as a serially concatenated code and applies independent decoding to the component codes \last{sequentially} across its dimensions. \last{Much later, t}he suitability of product code constructions for iterative decoding algorithms\cite{Berrou93} led to a very powerful class of codes\cite{Pyndiah98,Tanner,Li04,Feltstrom09,Pfister15}. \mcc{For} an overview of product codes and their variants\mcc{, we refer the reader to} \cite{Berrou05,hager17}. Usually, product codes are constructed with high-rate algebraic component codes, for which low-complexity \ac{SISO}\cite{Pyndiah98} or algebraic (e.g., bounded distance) \cite{Abramson,KoetterPC,Haeger2018} decoders are available. Specifically, product codes with \ac{SPC} component codes are considered in \cite{caire,rankin1,rankin2}, where the interest was mainly their performance and \mcc{their} weight enumerators.

In \cite{6125399}, \last{a} bridge between generalized concatenated codes and polar codes\cite{arikan2009channel,stolte2002rekursive} was established. In \cite{Coskun19:RM-Product}, the standard polar \ac{SCL} decoder is proposed for a class of product codes with Reed-Muller component codes, e.g., SPC \ed{codes whose length is a power of $2$} and/or extended Hamming component codes, \mcc{with non-systematic encoders}. Sizeable gains were observed with moderate list sizes over \ac{BP} decoding for short \ed{blocklengths} when the product codes were modified by \ed{introducing} a high-rate outer code.

In this paper, we show that \ac{SPC} product codes can be \ed{described using the tools} of polar codes based on generalized kernels \ed{\cite{PSL16,gabry_land,BGL20,benammar_land}}, where the frozen bit indices are chosen according to the constraints imposed by the product code construction. \mcc{Following this observation}, \ac{SC} decoding of \ac{SPC} product codes is introduced. A bridge between the original decoding algorithm of product codes, which is referred to as Elias' decoder\cite{Elias:errorfreecoding54},\footnote{By Elias' decoder, we refer to \last{the} decoding algorithm that treats the product code as a serially concatenated block code, where the decoding is performed starting from the component codes of the first dimension, up to those of the last dimension, in a one-sweep fashion.} and the \ac{SC} decoding algorithm is established for \ac{SPC} product codes over \last{the} \ac{BEC}. \last{Further,} the block error probability of \ac{SPC} product codes is upper bounded \last{via} the union bound under both decoding algorithms. A comparison between Elias' decoding and \ac{SC} decoding of \ac{SPC} product codes is \last{also} provided in terms of block error probability, proving that \ac{SC} decoding yields a \ed{probability of error that does not exceed the one of} Elias' decoding. \ed{The analysis of \ac{SC} decoding is extended to general \ac{BMS} channels.} Finally, \ac{SCL} decoding \cite{tal15} of \mcc{product} codes is introduced to overcome the significant performance gap of \ac{SC} decoding to \ed{the block error probability under \ac{ML} decoding (estimated through Poltyrev's \ac{TSB}\cite{poltyrev94})}. The performance improvement is significant, i.e., \ac{SCL} decoding \ed{yields a block error probability that is} below \ed{the} \ac{TSB}\ed{,} even for small list sizes\ed{,} for the analyzed code, \ed{delivering} a lower error probability compared to \ac{BP} decoding (especially at low error rates). 
In addition to the potential coding gain over \ac{BP} decoding, \ac{SCL} decoding enables a low-complexity decoding of the concatenation of the product code with a high-rate outer code as for polar codes \cite{tal15}. It is shown \ed{via simulations} that the concatenation provides remarkable gains over the product code alone. \ed{The gains would} not be possible under a \ac{BP} decoder which jointly decodes the outer code and the inner product code \cite{Coskun19:RM-Product}. \ed{We show examples where t}he resulting construction operates within \ed{$0.7$} dB of the \ac{RCU} bound\cite{Polyanskiy10:BOUNDS} with a moderate list size. \ed{From an application viewpoint, the performance gain with respect to BP decoding may be especially relevant for systems employing \ac{SPC} product codes with an outer error detection code (see, e.g., the IEEE~802.16 standard \cite{IEEE80216}). Moreover, we show that short \ac{SPC} product codes, concatenated with an outer \ac{CRC} code can outperform (under \ac{SCL}  decoding) 5G-NR \ac{LDPC} code with similar blocklength and dimension.}

\ed{By noticing that, for medium to short blocklengths, the \ac{SCL} decoder can approach the \ac{ML} decoder performance with a moderate list size, t}he analysis of the concatenated construction is addressed from a distance spectrum viewpoint. In \cite{caire}, a closed form expression is provided to compute the weight enumerator of $2$-dimensional \ac{SPC} product codes, relying on the MacWilliams identity for joint weight enumerators\cite{MacWilliams72}. In \cite{rankin1}, the closed form solution is extended to compute the input-output weight enumerator of $2$-dimensional \ac{SPC} product codes by converting the dual code into a systematic form. This method does not seem applicable for higher-dimensional constructions, as it is not trivial how to get to a systematic form of the dual code in such cases. In this work, the method in \cite{caire} is presented \last{using} a different approach, that avoids \ed{the use of} joint weight enumerators. This approach is \last{then} extended to accommodate the input-output weight enumerator of $2$-dimensional product codes, where one component code is an \ac{SPC} code. The method is used to compute the input-output weight enumerator of the exemplary short $3$-dimensional \ac{SPC} product code \last{as it can be seen as a $2$-dimensional product code, where one component code is an \ac{SPC} code}. By combining this result with the uniform interleaver approach, the average input-output weight enumerator of the concatenated code ensemble is computed, which is, then, used to compute some tight bounds on the block error probability \cite{di02,poltyrev94}, e.g., \ed{via} Poltyrev's \ac{TSB}.

The work is organized as follows. In Section \ref{sec:prelim}, we provide the preliminaries needed for the rest of the work. In Section \ref{sec:spcpc_as_pc}, we establish a bridge between \ac{SPC} product codes and the multi-kernel polar construction \cite{gabry_land,benammar_land}. The \ac{SC} decoding algorithm for \ac{SPC} product codes is described and analyzed over \ed{\ac{BMS} channels, with a particular focus on the \ac{BEC} and the \ac{B-AWGN} channel, in Section \ref{sec:sc}.} In Section \ref{sec:scl}, the \ac{SCL} decoding algorithm is described and the codes are analyzed from a \ac{ML} decoding point of view through their distance spectrum. Conclusions follow in Section \ref{sec:conc}.

%% file: preliminaries.tex
\section{Preliminaries}\label{sec:prelim}

\subsection{Notation}

In the following, lower-case bold letters are used for vectors, e.g., $\vecx = (x_1,x_2,\dots,x_n)$. The Hamming weight of $\vecx$ is $w_{\Ham}(\vecx)$. \mcc{When required, we} use $ x_a^b $ to denote the vector $ (x_a, x_{a+1}, \dots, x_b) $ where $ b > a $. Furthermore, we write $x_{a,m}^{b}$ to denote the subvector with indices $\{i\in[b]:a = i\mod m\}$, where $[b]$ denotes the set $\{1,2,\dots,b\}$. For instance, $x_{1,3}^9 = (x_1, x_4, x_7)$. In addition, $\vecx_{\sim i}$ refers to the vector where the \ed{element with index $i$} is removed, i.e., $ \vecx_{\sim i} = (x_1,x_2,\dots,x_{i-1},x_{i+1},\dots,x_{n})$. Component-wise \ed{addition} of two binary vectors \ed{in $\mathbbm{F}_2$} is denoted as $\vecx\oplus \vecy$. The $m$-digit multibase representation of a decimal number $a$ is denoted by $(a_1a_2\dots a_m)_{b_1b_2\dots b_m}$ and the conversion is done according to
\begin{equation}\label{eq:multibase}
	a = \sum_{i=1}^{m} a_i\prod_{j=i+1}^{m}b_j
\end{equation}
where $b_j$ is the base of the $j$-th digit $a_j$ with left-most digit being the most significant one with $0\leq a_j<b_j$. \ed{For example, the binary representation of a number is obtained by setting $b_j=2$, $j=1,\ldots,m$.}

Capital bold letters, e.g., $\boldsymbol{X}_{a\times b} = [x_{i,j}]$, are used for $a\times b$ matrices. \mcc{The subscript} showing the dimensions is omitted whenever the dimensions are clear from the context. Similarly, $\I_a$ refers to the $a\times a$ identity matrix. The Kronecker product of two matrices $\boldsymbol{X}$ and $\boldsymbol{Y}$ is
\begin{equation*}
\boldsymbol{X}\otimes \boldsymbol{Y} \triangleq \begin{bmatrix}
x_{1,1}\boldsymbol{Y} & x_{1,2}\boldsymbol{Y} & \dots \\ 
x_{2,1}\boldsymbol{Y} & x_{2,2}\boldsymbol{Y} & \dots \\
\vdots & \vdots &\ddots
\end{bmatrix}. \label{eq:Kronecker}
\end{equation*}
We define an $ab\times ab$ \emph{perfect shuffle} matrix \ed{\cite{R80}}, denoted as $\Per_{a,b}$, by the following operation
\begin{equation}
(x_1,x_2,\dots,x_{ab})\Per_{a,b} = (x_{1,b}^{ab},x_{2,b}^{ab},\dots,x_{b-1,b}^{ab},x_{b,b}^{ab}).
\label{def:permutation}
\end{equation}

We use capital letters, e.g., $X$, for \acp{RV} and lower-case counterparts, e.g., $x$, for their realizations. \ed{For random vectors, similar notation above is used, e.g., we use $X_a^b$ to denote the random vector $(X_a, X_{a+1},\ldots X_{b})$.} We denote a \ac{BMS} channel with input alphabet $ \mathcal{X} = \{0,1\} $, output alphabet $ \mathcal{Y} $, and transition probabilities \ed{(densities)} $ \W(y|x) $, $ x \in \mathcal{X} $, $ y \in \mathcal{Y} $ by $ \W : \mathcal{X} \rightarrow \mathcal{Y} $. \ed{For a given channel $\W$, let $I\left(\W\right)$ denote its mutual information with uniform inputs, which amounts to the capacity for any \ac{BMS} channel \cite{arikan2009channel}. Then}, $n$ independent uses of channel $\W$ are denoted as $\W^n : \mathcal{X}^n \rightarrow \mathcal{Y}^n$, with transition probabilites \ed{(densities)} $\W^n(\vecy|\vecx) = \prod_{i=1}^{n} \W(y_i|x_i)$. In addition, we define $\W_{\G} : \mathcal{X}^k \rightarrow \mathcal{Y}^n$ as the channel seen by the $k$-bit message, e.g., $\vecu$, to be encoded into $n$-bit $\vecx$ by a generator matrix $\G$\ed{, i.e., $\vecx = \vecu\G$}. In other words, the likelihood of message $\vecu$, encoded via $\G$, upon observing the channel output $\vecy$ is defined as
\begin{align}
    \W_{\G}(\vecy|\vecu)&\triangleq \W^n(\vecy|\vecu\G) \\
                        &= \W^n(\vecy|\vecx).
\end{align}
We write \ac{BEC}($ \epsilon_{\mcc{\mathrm{ch}}} $) to denote the \ac{BEC} with erasure probability $ \epsilon_{\mathrm{ch}} $. \ed{Here, the output alphabet is $\mathcal{Y} = \{0,1,?\}$, where $?$ denotes an erasure. The output of the \ac{BEC}($\epsilon_{\mathrm{ch}}$) is equal to the input (i.e., $y=x$) with probability $1-\epsilon_{\mathrm{ch}}$ and it is erased (i.e., $y=?$) with probability $\epsilon_{\mathrm{ch}}$.} We denote an \ac{SPC} code with blocklength $n$ by $\SPC_n$. For a given $(n,k)$ binary linear block code $\code$, its complete \ac{WEF} is
\begin{equation}
	\A_{\code}(\vecz) \triangleq \sum_{\vecx\in\code} \vecz^{\vecx}\label{eq:cwef}
\end{equation}
where $\vecz^{\vecx} \triangleq \ed{\prod_{i=1}^n z_i^{x_i}}$. \ed{Let $w_{\Ham}(\vecx)$ denote the Hamming weight of vector $\vecx$.} One sets $z_i = z$, $i=1,\dots,n$, to get \ed{(with slight abuse of notation)} the resulting \ac{WEF} of $\code$ as
\begin{equation*}
\A_{\code}(z) \triangleq \A_{\code}\left(z,z,\dots,z\right) = \sum_{i=0}^{n} A_iz^i
\end{equation*}
where $A_i$ is the number of codewords $\vecx\in\code$ of $w_{\Ham}(\vecx)=i$ (the sequence $A_0, A_1,\ldots,A_n$ is instead referred to as the weight enumerator of the code). The distinction between complete \acp{WEF} and \acp{WEF} should be clear from the different arguments. \ed{Finally,} we write $A_{\code}^{\mathrm{IO}}(x,z)$ as the \ac{IOWEF}, defined as
\begin{align}
    \A_{\code}^{\mathrm{IO}}(x,z) \triangleq \sum_{i=0}^{k}\sum_{w=0}^{n} A_{i,w}^{\mathrm{IO}}x^iz^w
\end{align}
\ed{where} $A_{i,w}^{\mathrm{IO}}$ is the number of codewords $\vecx\in\code$ of $w_{\Ham}(\vecu)=i$ and $w_{\Ham}(\vecx)=w$.

\subsection{Product Codes}

An $m$-dimensional  $(n,k,d)$ product code $\code$ is obtained by requiring that an $m$-dimensional array of bits satisfies a \ed{linear} code constraint along each axis \cite{Elias:errorfreecoding54}. \ed{More precisely, the information bits are arranged in an $m$-dimensional hypercube, where the length of dimension $i\in[m]$ is $k_i$. Then, the vectors in the $\ell$-th dimension are encoded via a linear (systematic) component code $\code_\ell$ with parameters $(n_\ell,k_\ell,d_\ell)$, where $n_\ell$, $k_\ell$, and $ d_\ell $ are its blocklength, dimension, and minimum Hamming distance, respectively.} The parameters of the resulting product code are \ed{\cite[Ch. 18, Sec. 2]{macwilliams}}
\begin{equation}
	n=\prod_{\ell=1}^m n_\ell, \quad k=\prod_{\ell=1}^m k_\ell, \quad \text{and} \quad \ed{d=\prod_{\ell=1}^m d_\ell}. \label{eq:parameters}
\end{equation}
\ed{The rate of the product code is
\[R\triangleq\frac{k}{n}=\prod_{\ell=1}^m R_\ell\]}
where $R_\ell$ is the rate of the $\ell$-th component code.

\subsubsection{Encoding}
A generator matrix of the $ m $-dimensional product code is \ed{\cite[Ch. 18, Sec. 2]{macwilliams}
\begin{equation}
    \G = \boldsymbol{G}_{1} \otimes \boldsymbol{G}_{2} \otimes \ldots \otimes \boldsymbol{G}_{m}\label{eq:generator_product_base}
\end{equation}
where $\boldsymbol{G}_{\ell}$} is the generator matrix of the $\ell$-th component code. Alternatively, we can define a generator matrix recursively as follows.
Let binary vectors $ \vecu $ and $ \vecx $ be the $k$-bit message to be encoded and the corresponding $n$-bit codeword, respectively, where the relation between them is given as $\vecx = \vecu\ed{\G^{[m]}}$, $\ed{\G^{[m]}}$ being the generator matrix of the product code with $m$ dimensions. We obtain $\ed{\G^{[m]}}$ recursively as
\begin{equation}
	\ed{\G^{[m]}} = \left(\I_{\ed{k^{[m-1]}}}\otimes\G_{\ed{m}}\right)\Per_{\ed{k^{[m-1]}},n_m}\left(\I_{n_m}\otimes\ed{\G^{[m-1]}}\right)
	\label{eq:generator_product}
\end{equation}
where $\ed{\G^{[0]}} \triangleq 1$ and $\ed{k^{[m-1]}} \triangleq \prod_{i=1}^{m-1}k_i$ with $\ed{k^{[0]}} = 1$ (observe that $\ed{k^{[m]}} = k$). \ed{We note also that $n = n^{[m]} \triangleq \prod_{i=1}^{m}n_i$ with $n^{[0]} = 1$. Fig. \ref{fig:transmission_recursive} depicts the encoding with \ac{SPC} product codes, where the encoding recursion is based on \eqref{eq:generator_product}.}
\begin{figure*}
	\begin{center}
		\includegraphics[width=\textwidth]{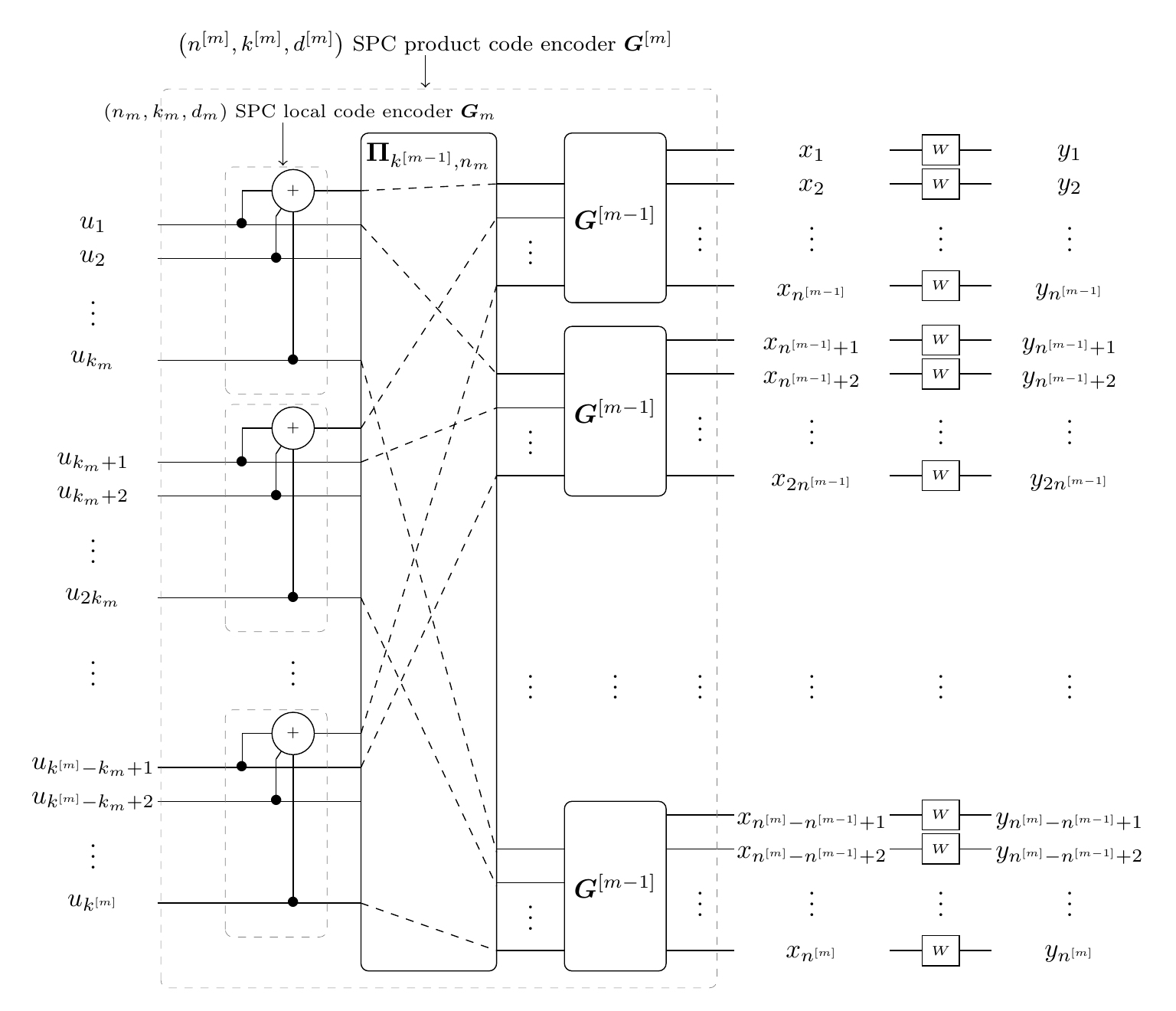}
	\end{center}
	\vspace{-0.5cm}
	\caption{\ed{Encoding using an $(n,k,d)$ \ac{SPC} product code with $m$ dimensions.}}
	\label{fig:transmission_recursive}
\end{figure*}

\ed{To see the relation between $\G$ and $\G^{[m]}$, we write
    \begin{align}
	    \ed{\G^{[m]}}
	    &= \left(\I_{\ed{k^{[m-1]}}}\otimes\G_{m}\right)\left(\ed{\G^{[m-1]}}\otimes\I_{n_{m}}\right)\Per_{\ed{n^{[m-1]}},n_{m}}\label{eq:generator_equivalence1}\\
	    &\!\!\!\!\!\!\!\!\!\!= \left(\ed{\G^{[m-1]}}\otimes\G_{m}\right)\Per_{\ed{n^{[m-1]}},n_{m}}\label{eq:generator_equivalence2} \\
	    &\!\!\!\!\!\!\!\!\!\!= \left(\G_{1}\otimes\G_{2}\otimes\ldots\otimes\G_{m}\right)\prod_{i=1}^{m}\left(\Per_{\ed{n^{[i-1]}},n_{i}}\otimes\I_{\nicefrac{\ed{n^{[m]}}}{\ed{n^{[i]}}}}\right)\label{eq:generator_equivalence5}
    \end{align}
    where \eqref{eq:generator_equivalence1} follows from applying the identity
    \[\Per_{\ed{k^{[m-1]}},n_m}\left(\I_{n_m}\otimes\ed{\G^{[m-1]}}\right) = \left(\ed{\G^{[m-1]}}\otimes\I_{n_{m}}\right)\Per_{\ed{n^{[m-1]}},n_{m}}\]
    \own{and \eqref{eq:generator_equivalence2} from the mixed-product identity. Then, \eqref{eq:generator_equivalence5} follows by re-writing $\G^{[m-1]}$ through \eqref{eq:generator_product} and by applying similar steps, recursively}. By noting the fact that the product of an arbitrary number of permutation matrices yields another permutation matrix, we can conclude that $\G$ and $\G^{[m]}$ are equivalent up to a column permutation for all $m\geq 1$.}
\subsubsection{Distance Spectrum}

Although the characterization of the complete distance spectrum of a product code is still an open problem even for the case where the distance spectrum of its component codes is known\mc{\cite{Tolhuizen02,El-Khamy05:conf,El-Khamy05}}, the minimum distance multiplicity is known~\ed{\cite[Theorem 3]{TOLHUIZEN92}} and equal to 
\[ A_{\ed{d^{[m]}}}=\prod_{\ell=1}^m A^{(\ell)}_{d_\ell}. \]
\mcc{Here, $ A^{(\ell)}_{d_\ell} $ is the minimum distance multiplicity of the $ \ell $-th component code.}

As a final remark, note that product codes can be \mcc{seen as a special class of generalized \ac{LDPC} codes} \cite{Tanner}. \ac{SPC} product codes, in particular, are a special class of (left-regular) \ac{LDPC} \cite{ldpc}, whose bipartite graph representation has girth $8$.

\subsection{Polar Codes}
\label{sec:polar}
Polar codes were \mcc{shown to be} the first class of provably capacity-achieving codes with low encoding and decoding complexity over any \ac{BMS} channel under \mcc{low-complexity} \ac{SC} decoding \cite{arikan2009channel}. In addition to the theoretical interest, polar codes concatenated with an outer \ac{CRC} code are very attractive from a practical viewpoint\ed{\cite[Ch. 5]{5G20}} thanks to their \mcc{excellent} performance \mcc{under} \ac{SCL} decoding \cite{tal15} in \mcc{the} short and moderate blocklength regime \cite{Coskun18:Survey}.

\ed{A transform matrix for a length $n = 2^m$ polar code is defined as $\GK_2^{\otimes m}$, where $\GK_2^{\otimes m}$ is the $m$-fold Kronecker product with $\GK_2^{\otimes 0}\triangleq 1$. Similar to the generator matrix \eqref{eq:generator_product} of product codes, an alternative construction of the transformation is possible. In this case, the $n \times n$ transform matrix $\boldsymbol{\mathsf{G}}^{[m]}$ is constructed recursively as
\begin{equation}
	\label{eq:generator_polar}
	\boldsymbol{\mathsf{G}}^{[m]} = \left(\I_{n/2}\otimes\GK_{2}\right)\Per_{n/2,2}\left(\I_2\otimes\boldsymbol{\mathsf{G}}^{[m-1]}\right)
\end{equation}
where the kernel at each recursion is fixed and it is defined as
\[
\GK_{2}\triangleq\begin{bmatrix}
1 & 0 \\ 
1 & 1
\end{bmatrix}
\]
and $\boldsymbol{\mathsf{G}}^{[0]} \triangleq 1$. Let $\vecw$ be any vector in $\mathcal{X}^n$ and it is mapped onto $\vecx$ as $\vecx = \vecw\boldsymbol{\mathsf{G}}^{[m]}$. Transition probabilities of the $i$-th bit-channel, a \emph{synthesized} channel with the input $w_i$ and the output $(\vecy, w_1^{i-1})$, are defined by
\begin{equation}
	\W^{(i)}_{\boldsymbol{\mathsf{G}}^{[m]}}(\vecy,w_1^{i-1}|w_i) \triangleq \sum_{w_{i+1}^n\in\mathcal{X}^{n-i}} \frac{1}{2^{n-1}}\W_{\boldsymbol{\mathsf{G}}^{[m]}}(\vecy|\vecw).
	\label{eq:bit_ML_polar}
\end{equation}
The channels $\W^{(i)}_{\boldsymbol{\mathsf{G}}^{[m]}}$, $i\in[n]$, polarize, i.e, the fraction of channels with $I\left(\W^{(i)}_{\boldsymbol{\mathsf{G}}^{[m]}}\right)>1-\delta$ goes to $I\left(\W\right)$ and the fraction of channels with $I\left(\W^{(i)}_{\boldsymbol{\mathsf{G}}^{[m]}}\right)<\delta$ to $1-I\left(\W\right)$ for any $\delta\in(0,1)$ as $n\rightarrow\infty$\cite[Theorem 1]{arikan2009channel}.}

\ed{A generator matrix $\G$} is obtained by removing the rows of \ed{$\boldsymbol{\mathsf{G}}^{[m]}$} with indices in $\mathcal A$, where $\mathcal A$ is the set containing the indices of the frozen bits. We refer to the matrix \ed{$\boldsymbol{\mathsf{G}}^{[m]}$} as the \ed{polar transform}, from which the desired polar code is derived. Encoding can be performed by multiplying the $k$-bit message $\vecu$ by $\G$, i.e., $\vecx = \vecu\G$. Equivalently, the encoding process can be described via matrix \ed{$\boldsymbol{\mathsf{G}}^{[m]}$}. In this case, an $n$-bit vector \ed{$\vecw$} has to be defined, where $\ed{w_i}=0$ for all $i\in \mathcal A$ and the remaining $k$ elements of \ed{$\vecw$} carry information. Encoding is then performed as $\vecx=\ed{\vecw\boldsymbol{\mathsf{G}}^{[m]}}$.

It was already mentioned in \cite{arikan2009channel} that generalizations of polar codes \mc{are} possible by choosing different kernels than $\GK_{2}$ \ed{and those kernels can even} be mixed. Later, the conditions for polarizing kernels \mc{were} provided in \cite{KSU10} and corresponding error exponents \mc{were} derived. Then, \cite{benammar_land} extended the error exponent derivation \mc{to} the constructions mixing kernels as \cite{arikan2009channel} suggested\mc{, while} \cite{gabry_land} provided \mc{examples of} constructions using this approach, namely multi-kernel polar codes. In the following, we \ed{study the relations between \ac{SPC} product codes and multi-kernel polar codes. In particular, we show how the kernels and the frozen bits can be chosen so that the tools of multi-kernel polar codes can be used in the description of \ac{SPC} product codes.}

%% file: spc_pc_as_polar_code.tex
\section{\ed{Relations Between Single Parity-Check Product Codes and Multi-Kernel Polar Codes}}\label{sec:spcpc_as_pc}

\ed{We consider $(n_\ell\times n_\ell)$ kernels $\GK_{n_\ell}$, where $n_\ell\geq 2$, $\ell\in[m]$,} of the form
\begin{equation}
	\GK_{n_\ell} = \begin{bmatrix}
	1 & 0 & \dots & 0 \\
	1 & & & \\
	\vdots & & \I_{k_\ell} & \\
	1 & & &
	\end{bmatrix}
	\label{eq:kernel}
\end{equation}
\ed{with $k_\ell = n_\ell - 1$.} Similar to \eqref{eq:generator_polar}, \ed{an $n^{[m]} \times n^{[m]}$ transform matrix $\boldsymbol{\mathsf{G}}^{[m]}$ is obtained, recursively,} as
\ed{\begin{equation}
	\boldsymbol{\mathsf{G}}^{[m]} = \left(\I_{n^{[m-1]}}\otimes\GK_{n_{m}}\right)\Per_{n^{[m-1]},n_{m}}\left(\I_{n_{m}}\otimes\boldsymbol{\mathsf{G}}^{[m-1]}\right)
	\label{eq:bigGpolar}
\end{equation}
where $\G^{[0]} \triangleq 1$. The proof of the following lemma is given as appendix.
\begin{lemma}\label{lemma:polarization}
    The multi-kernel construction \eqref{eq:bigGpolar}, with a sequence of kernels of the form \eqref{eq:kernel}, polarizes. More formally, the fraction of channels with $I\left(\W^{(i)}_{\boldsymbol{\mathsf{G}}^{[m]}}\right)>1-\delta$ goes to $I\left(\W\right)$ and the fraction of channels with $I\left(\W^{(i)}_{\boldsymbol{\mathsf{G}}^{[m]}}\right)<\delta$ goes to $1-I\left(\W\right)$ for any $\delta\in(0,1)$ as $m\rightarrow\infty$. This holds for any  sequence $n_1, n_2,\ldots$ where each $n_i$ satisfies $2\leq n_i<\infty$.
\end{lemma}}

\ed{Note that Lemma \ref{lemma:polarization}, upon proving that the rate of convergence is positive (which can be computed via \cite[Theorem 2]{benammar_land} after fixing the kernels), implies that a multi-kernel polar code based on the kernels of the form \eqref{eq:kernel} achieves capacity for general \ac{BMS} channels. In the following, however, we provide a selection procedure for the frozen bits yielding an \ac{SPC} product code, which does not take into account the quality of the synthesized channels. This hinders the possibility to achieve capacity for the \ac{SPC} product codes under \ac{SC} decoding.}

Recall the multibase representation \lasttt{\eqref{eq:multibase}} of a decimal number $i$, denoted by $(i_1i_2\dots i_m)_{n_1n_2\dots n_m}$. Then, \mcc{the} generator matrix $\ed{\G^{[m]}}$ is obtained by choosing set $\mathcal{A}_\text{PC}\subset [n]$ of frozen bits as
\begin{align}
	\mathcal{A}_\text{PC} = \ed{[n]\setminus\left\{i+1\in[n]: i_j \neq 0, \quad\forall j=1,2,\dots,m \right\}}.
	\label{eq:frozen_set}
\end{align}
Note that encoding can be done either by using \eqref{eq:generator_product} as $\vecx = \vecu\ed{\G^{[m]}}$, or by using \eqref{eq:bigGpolar} as $\vecx = \ed{\vecw\boldsymbol{\mathsf{G}}^{[m]}}$ with $\ed{w_i} = 0$ for all $i\in\mathcal{A}_\text{PC}$ and the remaining positions are allocated for the information bits as for polar codes (see Section \ref{sec:polar}). In other words, \eqref{eq:bigGpolar} generalizes \eqref{eq:generator_polar} to generate the mother code for multi-kernel polar codes generated by $m$ kernels $\GK_{n_{\ell}}$ in dimensions \ed{$\ell\in[m]$. Note that \eqref{eq:bigGpolar} recovers \eqref{eq:generator_polar} by setting $n_\ell = 2$ for all $\ell\in[m]$.}
\begin{example}
	Consider \mcc{the} $(3\times 3)$ kernel\mcc{s} 
	\[\mcc{\GK_{n_1} = \GK_{n_2} =} \begin{bmatrix}
	1 & 0 & 0 \\
	1 & 1 & 0 \\
	1 & 0 & 1
	\end{bmatrix}.\]
	We construct $\ed{\boldsymbol{\mathsf{G}}^{[2]}}$ by using \eqref{eq:bigGpolar}, i.e.,
	\[\ed{\boldsymbol{\mathsf{G}}^{[2]}} = \begin{bmatrix}
	1 & 0 & 0 & 0 & 0 & 0 & 0 & 0 & 0 \\
	1 & 0 & 0 & 1 & 0 & 0 & 0 & 0 & 0 \\
	1 & 0 & 0 & 0 & 0 & 0 & 1 & 0 & 0 \\
	1 & 1 & 0 & 0 & 0 & 0 & 0 & 0 & 0 \\
	1 & 1 & 0 & 1 & 1 & 0 & 0 & 0 & 0 \\
	1 & 1 & 0 & 0 & 0 & 0 & 1 & 1 & 0 \\
	1 & 0 & 1 & 0 & 0 & 0 & 0 & 0 & 0 \\
	1 & 0 & 1 & 1 & 0 & 1 & 0 & 0 & 0 \\
	1 & 0 & 1 & 0 & 0 & 0 & 1 & 0 & 1
	\end{bmatrix}.
	\]
	Then, the generator matrix $\ed{\G^{[2]}}$ is constructed as
	\[\ed{\G^{[2]}} = \begin{bmatrix}
	1 & 1 & 0 & 1 & 1 & 0 & 0 & 0 & 0 \\
	1 & 1 & 0 & 0 & 0 & 0 & 1 & 1 & 0 \\
	1 & 0 & 1 & 1 & 0 & 1 & 0 & 0 & 0 \\
	1 & 0 & 1 & 0 & 0 & 0 & 1 & 0 & 1
	\end{bmatrix}\]
	by removing the rows, with indices given by \eqref{eq:frozen_set}, i.e., $\mathcal{A}_\text{PC} = \{1,2,3,4,7\}$, \mcc{as depicted} in Fig. \ref{fig:size9_polar_construction_PC}(a). Equivalently, \mcc{the} generator matrix $\ed{\G^{[2]}}$ can be formed by \mcc{using \eqref{eq:generator_product} after} removing the first rows of the kernels to get \mcc{the generator} matrices $\G_{\ed{1}}$ and $\G_{\ed{2}}$ defining \ac{SPC} \mcc{component} codes, i.e.,
	\[\G_{\ed{1}} = \G_{\ed{2}} = \begin{bmatrix}
	1 & 1 & 0 \\
	1 & 0 & 1
	\end{bmatrix}.\]
	\begin{figure}    
		\centering
		\begin{subfigure}{0.5\textwidth}
			\includegraphics[width=1\textwidth]{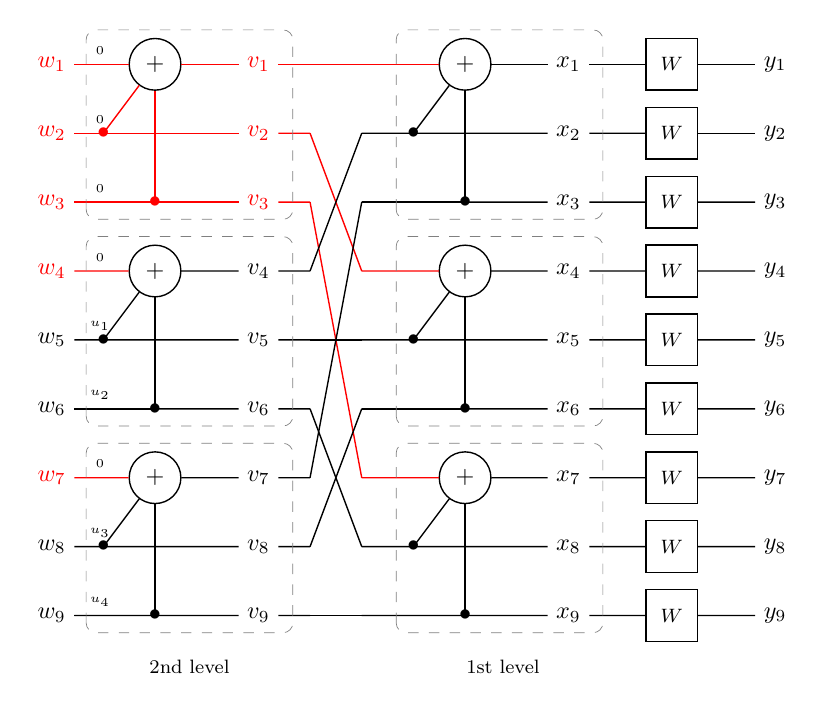}
			\vspace{-0.5cm}
			\caption{}
			%	\label{fig:size9_polar_construction_PCa}
		\end{subfigure}
		\begin{subfigure}{0.5\textwidth}
			\includegraphics[width=1\textwidth]{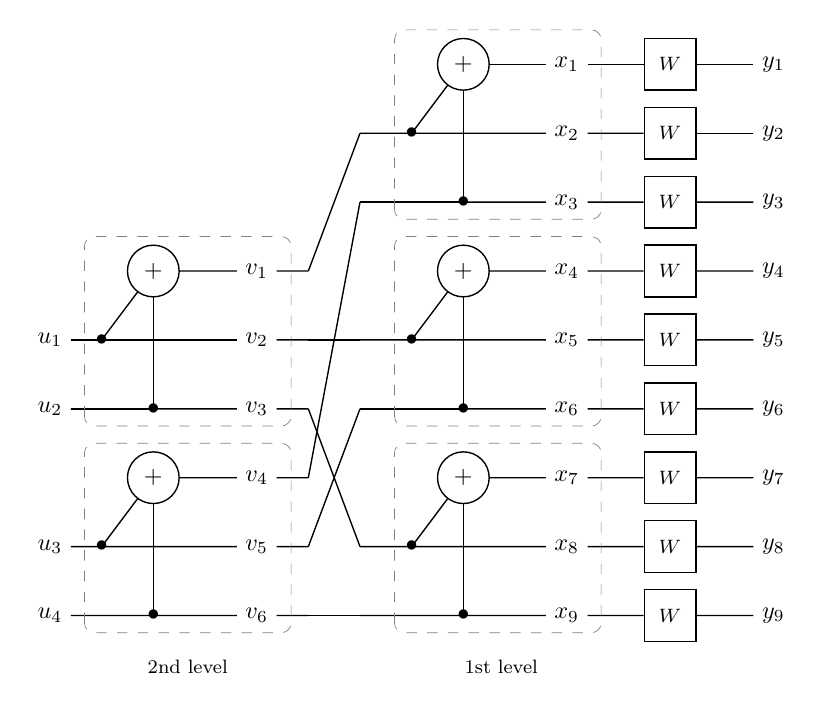}
			\vspace{-0.5cm}
			\caption{}
		\end{subfigure}
		\caption{(a) \ed{illustrates} how to choose the frozen bits \ed{to obtain} a SPC product code, where red edges show those carrying frozen bit values and red variables \ed{are set to} $0$. In (b), the frozen bits and the corresponding edges \ed{are removed, providing a graphical representation of a} $2$-dimensional SPC product code \ed{as in Fig. \ref{fig:transmission_recursive}}.}
		\label{fig:size9_polar_construction_PC}
	\end{figure}
	\label{ex:construction}
\end{example}

%% file: sc.tex
\section{Successive Cancellation Decoding of Single Parity-Check Product Codes}
\label{sec:sc}
Consider transmission over a \ac{BMS} channel $ \W $ using \mcc{an} $m$-dimensional $(n,k)$ \mcc{systematic} \ac{SPC} product code $\code$ with generator matrix $\ed{\G^{[m]}}$ \mcc{as in \eqref{eq:generator_product}, with component code generator matrices $\G_{\ed{m}}$ obtained by removing the first row of the kernels of the form \eqref{eq:kernel} (see Example \ref{ex:construction})}. Assume now that one is interested in the likelihood of $u_i$ upon observing \ed{the channel output $ \vecy \in \mathcal{Y}^n$} and given the knowledge of $u_1^{i-1}$
\begin{equation}
	\W^{(i)}_{\ed{\G^{[m]}}}(\vecy,u_1^{i-1}|u_i) \triangleq \sum_{u_{i+1}^k\in\mathcal{X}^{k-i}} \frac{1}{2^{k-1}}\W_{\ed{\G^{[m]}}}(\vecy|\vecu).
	\label{eq:bit_ML}
\end{equation}
\ed{Assume further that we interpret the \ac{SPC} product code by the multi-kernel polar code perspective discussed in Sec. \ref{sec:spcpc_as_pc} as depicted in Fig. \ref{fig:size9_polar_construction_PC}(a). The evaluation of \eqref{eq:bit_ML} entails the use of the knowledge of \emph{all} frozen bits, i.e., all $w_j$ with $j\in\mathcal{A}_{\mathrm{PC}}$, including the ones with indices larger than the bit index under consideration, e.g., the knowledge of $w_7 = 0$ to compute \eqref{eq:bit_ML} for $i=1$ or $i=2$.\footnote{\ed{In fact, note that this suboptimality is encountered also in polar codes for the decoding of the $i$-th bit \ed{$w_i$} because of the assumption that the future frozen bits $\{\ed{w_j}: j\in\mathcal{A}, i<j\leq n\}$ are uniform \acp{RV} \ed{(see Fig. \ref{fig:size9_polar_construction_PC}(b))} although they are deterministic and known.}} \ac{SC} decoding allows evaluating \eqref{eq:bit_ML} with good accuracy. The recursive operation of the \ac{SC} decoder can be easily described by means of the representation in Fig. \ref{fig:size9_polar_construction_PC}(b).}

\ac{SC} decoding follows the schedule in \cite{arikan2009channel,stolte2002rekursive} for polar codes. Explicitly, decision $\vecuhat$ is made successively as
\begin{equation}
\ed{\hat{u}_i(\vecy,\hat{u}_1^{i-1})} = \left\{\begin{array}{lll}
0 &\text{if}& \W^{(i)}_{\ed{\G^{[m]}}}(\vecy,\hat{u}_1^{i-1}|0) \geq \W^{(i)}_{\ed{\G^{[m]}}}(\vecy,\hat{u}_1^{i-1}|1) \\
1 &\text{if}& \text{otherwise}
\end{array}
\right.
\label{eq:decision2}
\end{equation}
for $i=1,\dots,k$ by approximating $\W^{(i)}_{\ed{\G^{[m]}}}(\vecy,\ed{u}_1^{i-1}|u_i)$ \ed{recursively as follows: Let $\vecy^{[m]}\triangleq y_1^{n^{[m]}}$ be the channel output vector for the $m$-dimensional \ac{SPC} product code. Assume $\vecy^{[m]}$ to be partitioned into $n_m$ blocks of length-$n^{[m-1]}$, where the $j$-th block is denoted as $\vecy_j^{[m]}\triangleq y_{jn^{[m-1]}+1}^{(j+1)n^{[m-1]}}$, $j=0,1,\ldots,k_m$ (see Fig. \ref{fig:transmission_recursive}). Then, the recursion to compute the likelihoods is} \own{given as \eqref{eq:recurs_metric} at the top of the next page and it is} continued down to $\W_{1}^{(1)}(y_i|x_i) \triangleq \W(y_i|x_i)$.
\begin{figure}
	\begin{center}
		\includegraphics[width=0.9\columnwidth]{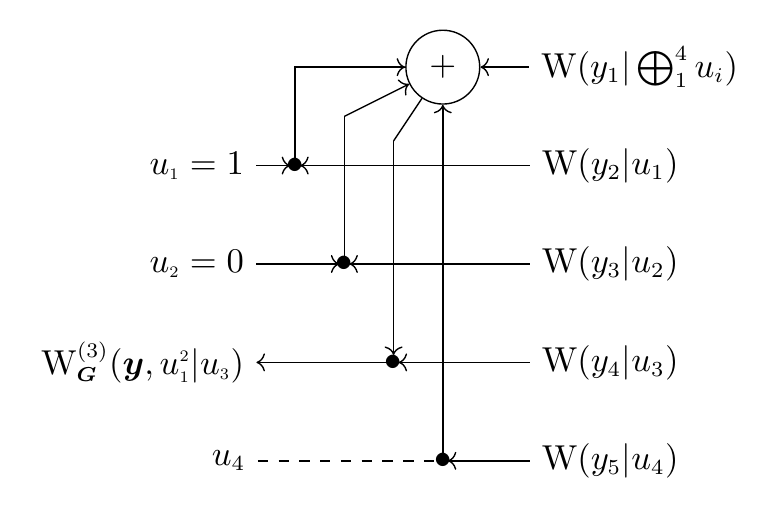}
	\end{center}
	\vspace{-0.5cm}
	\caption{\ed{The decoding graph of a $(5,4)$ \ac{SPC} code illustrating the \ac{SC} decoding equation \eqref{eq:recurs_metric} for bit $u_3$ with $j = 0$ and $t = 3$. Note that the arrows denote the flow of the quantities on the graph. The filled circles denote the equality constraints while the summation sign in the circle denote a parity check constraint. The knowledge of the bits $u_{1}^{2} = (1,0)$ is available while the upcoming bit $u_{4}$ is assumed to be uniformly distributed, illustrated with a dashed line.}}
	\label{fig:local_dec_sc}
\end{figure}
\begin{figure*}[t]
%	\begin{strip}
		\footnotesize
        \begin{align}
	        &\W^{(i)}_{\ed{\G^{[m]}}}\left(\ed{\vecy^{[m]}},\ed{u}_1^{i-1}|u_i\right) \approx \sum_{u_{i+1}^{(j+1)k_{m}}}\frac{1}{2^{k_{m}-1}}\W^{(j+1)}_{\ed{\G^{[m-1]}}}\left(\ed{\vecy_0^{[m]}},\bigoplus_{i^\prime = 1}^{k_m} \ed{u}_{i^\prime,k_m}^{jk_m}\Big|\bigoplus_{i^\prime = 1}^{k_m} u_{jk_m+i^\prime}\right)\prod_{j^\prime=1}^{k_m}\W^{(j+1)}_{\ed{\G^{[m-1]}}}\left(\ed{\vecy_{j^\prime}^{[m]}},\ed{u}_{j^\prime,k_m}^{jk_m}\Big|u_{jk_m+j^\prime}\right)\label{eq:recurs_metric1} \\
	        &= \!\prod_{j^\prime=1}^{t}\!\W^{(j+1)}_{\ed{\G^{[m-1]}}}\!\left(\vecy_{j^\prime}^{[m]},\ed{u}_{j^\prime,k_m}^{jk_m}\Big|u_{jk_m+j^\prime}\right)\!\sum_{u_{i+1}^{(j+1)k_{m}}}\!\frac{1}{2^{k_{m}-1}}\W^{(j+1)}_{\ed{\G^{[m-1]}}}\!\left(\vecy_{0}^{[m]},\bigoplus_{i^\prime = 1}^{k_m} \ed{u}_{i^\prime,k_m}^{jk_m}\Big|\bigoplus_{i^\prime = 1}^{k_m} u_{jk_m+i^\prime}\!\right)\!\prod_{j^\prime= t+1 }^{k_m}\!\W^{(j+1)}_{\ed{\G^{[m-1]}}}\!\left(\vecy_{j^\prime}^{[m]},\ed{u}_{j^\prime,k_m}^{jk_m}\Big|u_{jk_m+j^\prime}\!\right) \label{eq:recurs_metric}
        \end{align}
		where $j = \lfloor\frac{i-1}{k_m}\rfloor$ and $t = [(i-1)\!\!\mod k_m] +1$.
		\newline\decoRule
%	\end{strip}
\end{figure*}

\ed{To gain insight on \eqref{eq:recurs_metric1} and \eqref{eq:recurs_metric}, let us consider the simple case of a length-$5$ \ac{SPC} code with generator matrix $\G$. As illustrated in Fig. \ref{fig:local_dec_sc}, suppose that we are  interested in the likelihood $\W^{(3)}_{\G}(\vecy,u_1^{2}|u_3)$, for every $u_3\in\{0,1\}$, by assuming that the previous bits are given as $u_1^2 = (1,0)$. Using \eqref{eq:recurs_metric}, the computation is performed as
\begin{align}
	\W^{(3)}_{\G}(\vecy,u_1^{2}|u_3) &= \prod_{j^\prime=1}^{3}\W_{1}^{(1)}\left(y_{j^\prime+1}|u_{j^\prime}\right)\cdot\\
	&\sum_{u_{4}}\!\frac{1}{2^{3}}\W_{1}^{(1)}\!\left(y_{1}| u_1\!\oplus \!u_2\!\oplus\! u_3\!\oplus\! u_4\!\right)\!\W_{1}^{(1)}\!\left(y_{5}|u_{4}\!\right) \\
	&=\W\left(y_{2}|1\right)\W\left(y_{3}|0\right)\W\left(y_4|u_3\right)\cdot\\
	&\sum_{u_{4}}\frac{1}{2^{3}}\W\left(y_{1}| 1\oplus u_3\oplus u_4\right)\W\left(y_{5}|u_{4}\right) \label{eq:recurs_metric_SPC}
\end{align}
where \eqref{eq:recurs_metric_SPC} follows by plugging in the values of bits $u_1^2$ and noting that the recursion ends at $\W_{1}^{(1)}(y_i|x_i) = \W(y_i|x_i)$, where $x_1 = \bigoplus_{i=1}^4 u_i$ and $x_{i+1} = u_i$, $i=1,2$.\footnote{\ed{Note also that \eqref{eq:recurs_metric_SPC} computes the desired probability exactly since there is no future frozen bit in the polar code representation of an \ac{SPC} code while decoding any information bit. This turns out to be an approximation for the case of an \ac{SPC} product code, i.e., when $m\geq 2$.}}}

Over the \ac{BEC}, ties are not broken towards any decision by revising \eqref{eq:decision2} as
\begin{equation}
\ed{\hat{u}_i(\vecy,\hat{u}_1^{i-1})} = \left\{\begin{array}{lll}\vspace{1mm}
0 &\text{if}& \W^{(i)}_{\ed{\G^{[m]}}}(\vecy,\hat{u}_1^{i-1}|0) > \W^{(i)}_{\ed{\G^{[m]}}}(\vecy,\hat{u}_1^{i-1}|1) \\\vspace{1mm}
\last{?} &\text{if}& \W^{(i)}_{\ed{\G^{[m]}}}(\vecy,\hat{u}_1^{i-1}|0) = \W^{(i)}_{\ed{\G^{[m]}}}(\vecy,\hat{u}_1^{i-1}|1) \\
1 &\text{if}& \mcc{\text{otherwise.}}
\end{array}
\right.
\label{eq:decision_BEC}
\end{equation}
\ed{for $i=1,\dots,k$}. A block error event occurs if $ \vecuhat \neq \vecu $\ed{, where \[\vecuhat\triangleq (\hat{u}_1(\vecy),\hat{u}_2(\vecy,\hat{u}_1),\ldots,\hat{u}_k(\vecy,\hat{u}_1^{k-1})).\]}

The event that the decoding of $u_i$ is erroneous under \ac{SC} decoding, for which the knowledge of $u_1^{i-1}$ is available at the decoder via a genie, is defined as
\begin{equation*}
	\mathcal{B}_{\SC,i}^{\mathrm{GA}} \triangleq \{(\vecu, \vecy)\in\mathcal{X}^k\times\mathcal{Y}^n: \hat{u}_i(\vecy,u_1^{i-1}) \neq u_i\}.
\end{equation*}
Then, the block error event of the \ac{SC} decoding is equal to that of the genie-aided \ac{SC} decoding as stated in the following lemma. The proof is skipped as it can be easily \lasttt{derived} from\ed{\cite[Lemma 1]{Mori:2009SIT}}.
	\begin{lemma}
		\label{lem:block_error}
		The block error event for the \ac{SC} decoder satisfies
		\[\mathcal{E}_{\SC} = \bigcup_{i=1}^{k} \mathcal{B}_{\SC,i}^{\mathrm{GA}}.\]
	\end{lemma}
\ed{The block error probability under \ac{SC} decoding, denoted by $P_{\SC}$, is defined as $P_{\SC}\triangleq P(\mathcal{E}_{\SC})$ and it is bounded as
\begin{equation}
	\label{eq:SC_bound}
	\max_{i=1,\ldots,k} P(\mathcal{B}_{\SC,i}^{\mathrm{GA}}) \leq P_{\SC} \leq\sum_{i=1}^{k} P(\mathcal{B}_{\SC,i}^{\mathrm{GA}})
\end{equation}
where the upper bound follows from the straightforward application of the union bound.}
\begin{remark}
	Assume now that one is interested in the likelihood of $u_i$ upon observing $\vecy$ (by not imposing any order in the decoding)
	\begin{equation}
		\W^{(i)}_{\ed{\G^{[m]}}}(\vecy|u_i) \triangleq \sum_{\vecu_{\sim i}\in\mathcal{X}^{k-1}} \frac{1}{2^{k-1}}\W_{\ed{\G^{[m]}}}(\vecy|\vecu).
		\label{eq:bit_ML_elias}
	\end{equation}
	Observe that \eqref{eq:bit_ML_elias} corresponds to the bit-wise \ac{ML} function of \mc{the} information bit $u_i$. \ed{Computing \eqref{eq:bit_ML_elias} is hard in general}. Elias' decoder \cite{Elias:errorfreecoding54} tackles the problem as follows: the likelihoods of \mc{the} bits $u_i$ \eqref{eq:bit_ML_elias} are approximated starting from the first dimension up to the last one in a one-sweep fashion. \ed{Let $\vecu^{[m]}\triangleq u_1^{k^{[m]}}$ be the $k$-bit information to be encoded via $m$-dimensional \ac{SPC} product code and it is divided into $k^{[m-1]}$ blocks of length $k_m$, where $j$-th block is denoted as $\vecu_{j}^{[m]}\triangleq u_{jk_m+1}^{(j+1)k_m}$, $j=0,\ldots,k^{[m-1]}-1$ (see Fig. \ref{fig:transmission_recursive}).}
	For the computation, \eqref{eq:recurs_metric} is revised as
	\begin{figure*}[t]
%	\begin{strip}
		\footnotesize
        	\begin{align}
		        \W^{(i)}_{\ed{\G^{[m]}}}\left(\ed{\vecy^{[m]}}|u_i\right) &\approx \sum_{\ed{(\vecu_{j}^{[m]})_{\sim i}}}\frac{1}{2^{k_m-1}}\W^{(j+1)}_{\ed{\G^{[m-1]}}}\left(\ed{\vecy_0^{[m]}}\Big|\bigoplus_{i^\prime = 1}^{k_m} u_{jk_m+i^\prime}\right)\prod_{j^\prime=1}^{k_m}\W^{(j+1)}_{\ed{\G^{[m-1]}}}\left(\ed{\vecy_{j^\prime}^{[m]}}\Big|u_{jk_m+j^\prime}\right)\\
	        	&\ed{=\W^{(j+1)}_{\ed{\G^{[m-1]}}}\!\left(\ed{\vecy_{t}^{[m]}}\Big|u_{jk_m+t}\right)\sum_{\ed{(\vecu_{j}^{[m]})_{\sim i}}}\frac{1}{2^{k_m-1}}\W^{(j+1)}_{\ed{\G^{[m-1]}}}\left(\ed{\vecy_{0}^{[m]}}\Big|\bigoplus_{i^\prime = 1}^{k_m} u_{jk_m+i^\prime}\right)\prod_{j^\prime=1,j^\prime\neq t}^{k_m}\W^{(j+1)}_{\ed{\G^{[m-1]}}}\!\left(\ed{\vecy_{j^\prime}^{[m]}}\Big|u_{jk_m+j^\prime}\right)}
	            \label{eq:recurs_elias_metric}
        	\end{align}
		where $j = \lfloor\frac{i-1}{k_m}\rfloor$ and $t = [(i-1)\!\!\mod k_m] +1$.
		\newline\decoRule
%	\end{strip}
    \end{figure*}
	\own{\eqref{eq:recurs_elias_metric}}. On contrary to \ac{SC} decoding, the summation is over all information bits of the local SPC code in the $m$-th level except for $u_i$, i.e., over all $(\vecu_{j}^{[m]})_{\sim i}\in\{0,1\}^{k_m-1}$, which is computed for both values of $u_i$. In other words, the decoder does not make use of any decision on information bits to decode another one. This enables fully parallel computation of the likelihoods \eqref{eq:recurs_elias_metric} for each bit.
	
	\ed{To better illustrate the difference to the \ac{SC} decoding, we consider the Elias' decoding of bit $u_3$ in Fig. \ref{fig:local_dec_ed} for the case of $(5,4)$ \ac{SPC} code. In particular, we are interested in the likelihoods $\W^{(3)}_{\G}(\vecy|u_3)$, for every $u_3\in\{0,1\}$. Using \eqref{eq:recurs_elias_metric}, we compute
    \begin{align}
	    &\W^{(3)}_{\G}(\vecy|u_3) = \W_{1}^{(1)}\left(y_{4}|u_{3}\right)\cdot\\
	    &\,\,\,\,\sum_{\own{\vecu_{\sim 3}}}\frac{1}{2^{3}}\W_{1}^{(1)}\left(y_{1}| u_1\oplus u_2\oplus u_3\oplus u_4\right)\prod_{j^\prime = 1,\,j^\prime \neq 3}^{4}\W_{1}^{(1)}\left(y_{j^\prime+1}|u_{j^\prime}\right) \\
	    &=\!\W\!\left(y_4|u_3\!\right)\!\sum_{\own{\vecu_{\sim 3}}}\frac{1}{2^{3}}\W\!\left(y_{1}| 1\oplus u_3\oplus u_4\right)\!\own{\prod_{j^\prime = 1,\,j^\prime \neq 3}^{4}\!\W\left(y_{j^\prime+1}|u_{j^\prime}\right)}. \label{eq:recurs_metric_SPC_ed}
    \end{align}}
	\begin{figure}
	\begin{center}
		\includegraphics[width=0.9\columnwidth]{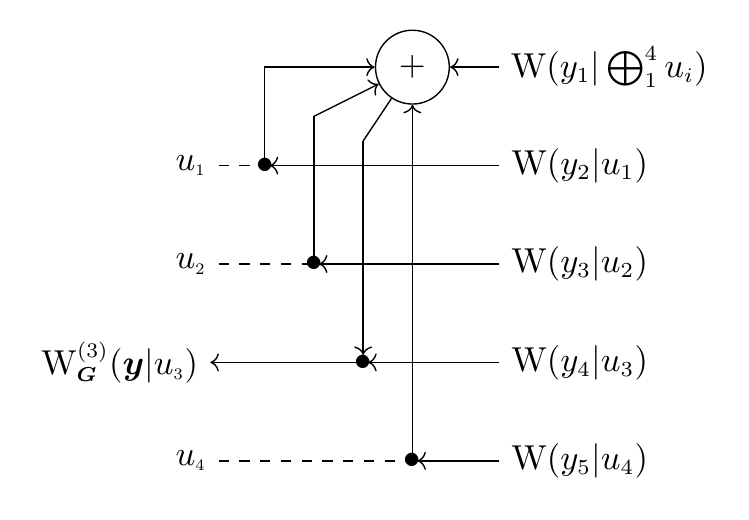}
	\end{center}
	\vspace{-0.5cm}
	\caption{\ed{The decoding graph of a $(5,4)$ \ac{SPC} code illustrating the Elias' decoding equation \eqref{eq:recurs_elias_metric} for bit $u_3$ with $j = 0$ and $t = 3$. Note that the arrows denote the flow of the quantities on the graph. The filled circles denote the equality constraints while the summation sign in the circle denote a parity check constraint. Unlike \ac{SC} decoding, bits $u_{\{1,2,4\}}$ are assumed to be uniformly distributed, illustrated with a dashed line. Hence, Elias' decoding of each bit $u_i$ can be run in parallel.}}
	\label{fig:local_dec_ed}
\end{figure}
	Elias' decoder, then, makes a decision \ed{as
	\begin{equation}
        \hat{u}_i^{(\mathrm{E})}(\vecy) = \left\{\begin{array}{lll}
        0 &\text{if}& \W^{(i)}_{\ed{\G^{[m]}}}(\vecy|0) \geq \W^{(i)}_{\ed{\G^{[m]}}}(\vecy|1) \\
        1 &\text{if}& \text{otherwise}
        \end{array}
        \right.
        \label{eq:decision_elias}
    \end{equation}
    for $i\in[k]$. Over the BEC, the decision function is modified as
    \begin{equation}
        \hat{u}_i^{(\mathrm{E})}(\vecy) = \left\{\begin{array}{lll}\vspace{1mm}
        0 &\text{if}& \W^{(i)}_{\ed{\G^{[m]}}}(\vecy|0) > \W^{(i)}_{\ed{\G^{[m]}}}(\vecy|1) \\\vspace{1mm}
        ? &\text{if}& \W^{(i)}_{\ed{\G^{[m]}}}(\vecy|0) = \W^{(i)}_{\ed{\G^{[m]}}}(\vecy|1) \\
        1 &\text{if}& \text{otherwise.}
        \end{array}
        \right.
        \label{eq:decision_Elias_BEC}
    \end{equation}}
    \ed{A block error event occurs if $ \vecuhat^{(\mathrm{E})} \neq \vecu $, where \[\vecuhat^{(\mathrm{E})}\triangleq (\hat{u}_1^{(\mathrm{E})}(\vecy),\hat{u}_2^{(\mathrm{E})}(\vecy),\ldots,\hat{u}_k^{(\mathrm{E})}(\vecy)).\]}
\end{remark}

%% file: analysis.tex
\subsection{Analysis over the Binary Erasure Channel}

In the following, we analyze the \ac{SC} decoder of \ac{SPC} product codes over the \ac{BEC}. We do so to gain a deeper understanding on the behavior of the \ac{SC} decoder when applied to the code construction under investigation. We start by analyzing the behavior of \ed{an $(n,n-1)$} \ac{SPC} code with generator matrix \ed{$\G$ over the \ac{BEC}($\epsilon_{\mathrm{ch}}$)}. We denote by $\ed{\epsilon^{(i)}_{\G}}$ the erasure probability for the $i$-th information bit after \ac{SC} decoding conditioned on the knowledge of the $i-1$ preceding information bits, $i=1,\ldots , n-1$. \ed{When the knowledge of the $i-1$ preceding information bits is available, then the decoding of the $i$-th information bit is successful either when $y_{i+1} = ?$ or there is no erasure in the subvector $(y_1,y_{i+2}^{n})$. Hence,} the relationship between the input-output erasure probabilities is given by
\ed{\begin{equation}
	\ed{\epsilon^{(i)}_{\G}}=\epsilon_{\mathrm{ch}}\left(1-\left(1-\epsilon_{\mathrm{ch}}\right)^{n-i}\right),\quad i\in [n-1].
	\label{eq:erasure_SPC}
\end{equation}}

Based on the relation given in \eqref{eq:erasure_SPC}, we proceed by bounding the performance of an $(n,k)$ \ac{SPC} product code $\code$, via \eqref{eq:SC_bound}. In particular, we can derive the erasure probability associated with the information bit $u_i$ of an \ac{SPC} product code under the genie-aided \ac{SC} decoding by iterating \eqref{eq:erasure_SPC}. \ed{More precisely, for $i \in [k]$, we have the recursion in $m$ as
\begin{equation}
	\epsilon^{(i)}_{\G^{[m]}}=\epsilon^{(j +1)}_{\G^{[m-1]}}\left(1-\left(1-\epsilon^{(j +1)}_{\G^{[m-1]}}\right)^{n_m-t}\right)
	\label{eq:erasure_SPC_PC}
\end{equation}
with $j = \lfloor\nicefrac{(i-1)}{k_m}\rfloor$ and $t = [(i-1)\!\!\mod k_m] +1$ where $k$ information bits are divided into $k^{[m-1]}$ blocks $\vecu_{j}^{[m]}$, $j\in[k^{[m-1]}]$, of size $k_m$ (see Fig. \ref{fig:transmission_recursive}).} Let $\epsilon_i$ be a shorthand for $\ed{\epsilon^{(i)}_{\G^{[m]}}}$, i.e., $\epsilon_i\triangleq\ed{\epsilon^{(i)}_{\G^{[m]}}}$ $i=1,\ldots,k$.\footnote{If the blocklength and the rate of an \ac{SPC} product code is given, then there exists a unique sequence of component codes satisfying the parameters \ed{\eqref{eq:parameters}}. For such a sequence, when the blocklengths of component codes are not equal, then a question is what decoding order should be adopted. The natural approach is to start the decoding from the lowest rate \ac{SPC} component code, i.e., to treat it as the component code in the first level as in Fig. \ref{fig:size9_polar_construction_PC}(b), because a code with a lower rate has a higher error-correction capability. \ed{This ordering has also been verified via numerical computation for an exemplary construction provided as Example \ref{sec:asymptotic1}, where a larger threshold is obtained if the decoding is performed in the reverse order of the component code rates.}}
\ed{Since the \ac{RHS} of \eqref{eq:erasure_SPC_PC} is monotonically increasing in the input erasure probability $\epsilon^{(j +1)}_{\G^{[m-1]}}$ and monotonically decreasing in $t$,} the largest bit erasure probability is equal to that of the first decoded information bit, i.e.,
\begin{equation}
    \epsilon_{\mathrm{max}}\triangleq \max_{i=1,\ldots,k} \epsilon_i = \epsilon_1.
\end{equation}
By rewriting \eqref{eq:SC_bound} in terms of $\mcc{\epsilon}_i$, we obtain
\begin{equation}
	\mcc{\epsilon}_{\mathrm{max}}\leq P_{\SC} \leq \sum_{i=1}^k \mcc{\epsilon}_i \label{eq:UB_Product}.
\end{equation}
A loose upper bound can be obtained by tracking only the largest erasure probability \mcc{for $i=1$}, i.e.,
\begin{equation}
	P_{\SC} \leq k\mcc{\epsilon}_{\mathrm{max}}.
	\label{eq:LooseUB_Product}
\end{equation}

\subsubsection{\GL{Comparison with Elias' Decoder}}

Remarkably, \eqref{eq:erasure_SPC_PC} also describes the evolution of the bit erasure probabilities under Elias' decoder \cite{Elias:errorfreecoding54} by setting $i=1$. The following lemma, together with Theorem \ref{theorem:comparison}, formalizes the relation between the error probability of \ac{SC} decoding and the one of the decoding algorithm proposed by Elias.

%\medskip

\begin{lemma}
	For an \ac{SPC} product code, the erasure probability of the first decoded information bit under \ac{SC} decoding is equal to the erasure probability of each information bit under Elias' decoding. \label{lem:worst_channel}
\end{lemma}
\begin{proof}
    \ed{Consider the input-output relation \eqref{eq:erasure_SPC} for an $(n,n-1)$ \ac{SPC} code under \ac{SC} decoding. In the case of Elias' decoding each bit is decoded in parallel without any knowledge of the information bits; therefore, each bit has the same erasure probability. This probability is obtained by setting $i=1$ in \eqref{eq:erasure_SPC} since the decoder does not have access to the knowledge of any information bit in this case. The same applies for the other decoding levels; hence,} the recursion \eqref{eq:erasure_SPC_PC} for $i=1$ also provides the erasure probability of any information bit under Elias' decoding.
\end{proof}
As a result \last{of Lemma \ref{lem:worst_channel}}, the bound \eqref{eq:LooseUB_Product} holds also for \mcc{the block error probability $P_{\E}$ under} Elias' decoding\mcc{, i.e., we have}
\begin{equation} \label{eq:bounds_ED}
\mcc{\epsilon}_{\mathrm{max}} < P_{\E} \leq k \mcc{\epsilon}_{\mathrm{max}}.
\end{equation}

%\medskip

By comparing \eqref{eq:LooseUB_Product} with the \ac{RHS} of \eqref{eq:bounds_ED}, the question on whether the two algorithms provide the same block error probability may arise. In the following, it is shown that the \ed{block error probability under \ac{SC} decoding is upper bounded by the block error probability under Elias' decoding}.

%\medskip

\begin{theorem} For an \GL{$(n,k)$} \ac{SPC} product code, the block error probabilities under \ac{SC} and Elias' decoding over the \ac{BEC}($ \epsilon $) satisfy
	\begin{equation}
		\ed{P_{\SC} \leq P_{\E}}.
	\end{equation}
	\label{theorem:comparison}
\end{theorem}
%\vspace{-2em}
\begin{proof}
	The $i$-th bit error event under Elias' decoder is
	\begin{equation*}\label{eq:elias_bit_error_event}
	\mathcal{B}_{\E,i} \triangleq \{(\vecu, \vecy)\in\mathcal{X}^k\times\mathcal{Y}^n: \hat{u}_i^{(\mathrm{E})}(\vecy) \neq u_i\} 
	\end{equation*}
with $ \hat{u}_i^{(\mathrm{E})}(\vecy) $ being the output of Elias' decoder for $u_i$. \ed{First, we} show that \ed{over the \ac{BEC}} $\mathcal{B}_{\SC,i}^{\mathrm{GA}}\subseteq\mathcal{B}_{\E,i}$. \ed{To this end, we write
	\begin{align}
		&\mathcal{B}_{\SC,i}^{\mathrm{GA}} = \{(\vecu, \vecy)\in\mathcal{X}^k\times\mathcal{Y}^n: \hat{u}_i^{(\mathrm{SC})}(\vecy,u_1^{i-1}) = ?\} \\
		&= \!\{\!(\vecu, \vecy)\!\in\!\mathcal{X}^k\times\mathcal{Y}^n\!:\! \hat{u}_i^{(\mathrm{SC})}(\vecy,u_1^{i-1}) = ?, \hat{u}_i^{(\mathrm{E})}(\vecy) = ?\}\label{eq:equality}\\
		&\subseteq \!\{\!(\vecu, \vecy)\!\in\!\mathcal{X}^k\times\mathcal{Y}^n\!:\! \hat{u}_i^{(\mathrm{E})}(\vecy) = ?\} = \mathcal{B}_{\E,i}\label{eq:equality2}
	\end{align}
	where \eqref{eq:equality} follows from the fact that an erasure at the output of genie-aided SC decoding of the $i$-th bit implies an erasure for its Elias' decoding. Combining \eqref{eq:equality2} with Lemma~\ref{lem:block_error} concludes the proof.}
\end{proof}
\ed{Fig. \ref{fig:sc_125_64_SPC_bec} illustrates the simulation results for the $3$-dimensional $(125,64)$ \ac{SPC} product code, obtained by iterating $(5,4)$ \ac{SPC} codes, over the \ac{BEC}. The results are provided in terms of \ac{BLER} vs. channel erasure probability $\epsilon_{\mathrm{ch}}$. The \ac{SC} decoding performance is compared to the performance under Elias' decoding. The former outperforms the latter. The tight upper bound on the \ac{SC} decoding, computed via the \ac{RHS} of \eqref{eq:UB_Product}, is also provided.}
\begin{figure}[t]
	\begin{center}
		\includegraphics[width=1\columnwidth]{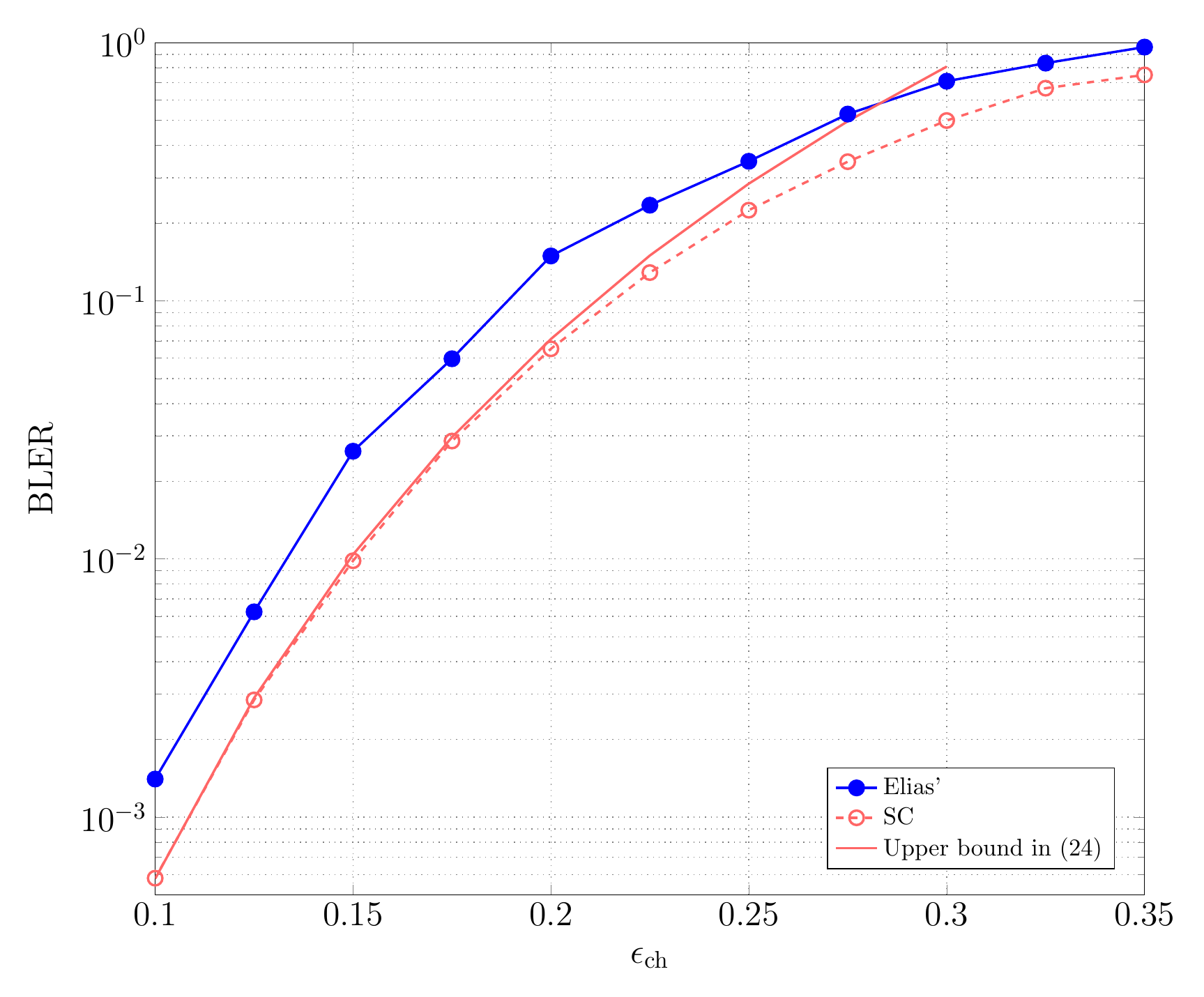}
	\end{center}
	\vspace{-0.5cm}
	\caption{\ed{\ac{BLER} vs. $\epsilon_{\mathrm{ch}}$  under \ac{SC} decoding for the $(125,64)$ product code, compared to Elias' decoding.}}\label{fig:sc_125_64_SPC_bec}
\end{figure}
\ed{\begin{remark}
    The inequality in Theorem \ref{theorem:comparison} can be made strict for an $m$-dimensional product code with $m>1$ and for which $\exists\ell\in\{1,\dots,m\}$ such that $n_\ell>2$. The proof is tedious as it requires the definition of erasure patterns that are resolvable by the \ac{SC} decoder while Elias' decoder fails. The general expression of such erasure patterns yields a complicated expression that we omit. In the following, we provide an example of such a pattern for the $(9,4)$ \ac{SPC} product code.
%\vspace{-2mm}
\end{remark}}
\begin{example}\label{example:comparison}
	Consider transmission \ed{over the \ac{BEC}} using the $ (9,4) $ product code with the received vector $ \vecy = \{0,?,?,?,0,?,0,0,0\}$, where $\mathcal{E}=\{2,3,4,6\}$. Under \ac{SC} decoding, the message is decoded correctly while Elias' decoding would fail to decode the $3$rd information bit \ed{as provided in Fig. \ref{fig:size9_polar_PC_ex}}.
	\begin{figure}    
		\centering
		\begin{subfigure}{0.5\textwidth}
			\includegraphics[width=1\textwidth]{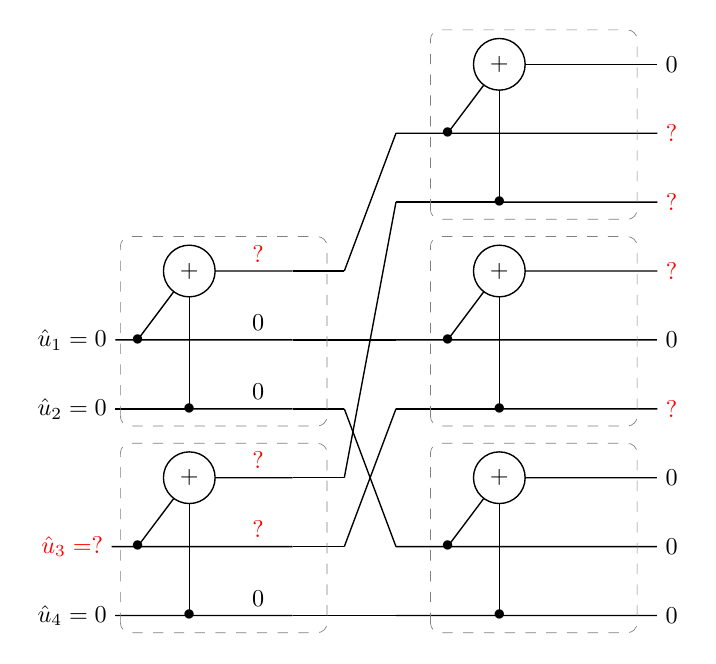}
			\vspace{-0.5cm}
			\caption{}
			%	\label{fig:size9_polar_construction_PCa}
		\end{subfigure}
		\begin{subfigure}{0.5\textwidth}
			\includegraphics[width=1\textwidth]{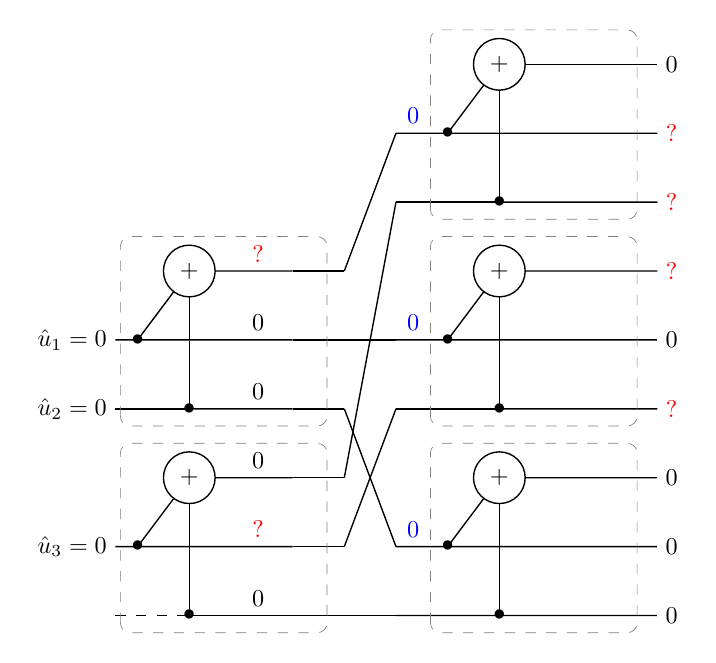}
			\vspace{-0.5cm}
			\caption{}
		\end{subfigure}
		\caption{\ed{(a) illustrates the output of Elias' decoding and (b) shows that $u_3$ indeed is correctly decoded by \ac{SC} decoding for this erasure pattern. Note that the blue labeled values denote the hard-decisions propagating from left to right in SC decoding.}}
		\label{fig:size9_polar_PC_ex}
	\end{figure}
\end{example}

\subsubsection{Asymptotic Performance Analysis}\label{subsec:asymptotic}

We consider now the asymptotic performance of \ac{SPC} product codes. More specifically, we analyze the error probability of a \emph{product code sequence}, defined by \ed{an ordered} sequence of component code sets
\[
	\mathscr{C}^{\ed{[m]}}=\left\{\code_1, \code_2, \ldots, \code_m\right\}
\] 
where we constrain $\left|\mathscr{C}^{\ed{[m]}}\right|=m$, i.e., where the number of component codes for the $m$-th product code in the sequence is set to $m$\ed{, and the component code rates satisfy $R_i\leq R_j$ for $i<j$}. \ed{We denote by $\code^{\ed{[m]}}$ the $m$-th product code in the sequence corresponding to the set $\mathscr{C}^{\ed{[m]}}$.} We aim at studying the behavior of the error probability as the dimension $m$ tends to infinity \ed{when the \ac{SC} decoding starts from $\code_1$ up to $\code_m$}. We remark  that, as $m$ changes, the component codes used to construct the product code are allowed to change, i.e., the sequence of product codes is defined by the set of component codes employed for each value of $m$.
Observe that the rate of $\code^{\ed{[m]}}$ may vanish as $m$ grows large, if the choice of the component codes forming the sets $\mathscr{C}^{\ed{[m]}}$ is not performed carefully.
We proceed by analyzing the limiting behavior in terms of \emph{block erasure thresholds} for different product code sequences \ed{with positive rates} under \ac{SC} decoding. Recall that we consider $m$-dimensional systematic \ac{SPC} product code constructions, where component code generator matrices $\G_{\ed{1}},\dots,\G_{\ed{m}}$ are obtained by removing the first rows of the kernels $\GK_{n_1},\dots,\GK_{n_m}$ of the form \eqref{eq:kernel}.
\begin{definition}\label{def:threshold}
	The \ac{SC} decoding block erasure threshold of a \ac{SPC} product code sequence defined by the component code \ed{sets $\mathscr{C}^{[1]},\ldots,\mathscr{C}^{[m]}$} is the largest channel erasure probability $\epsilon^\star$ for which the block error probability $P_{\SC}\left(\code^{\ed{[m]}}\right)$ converges to $0$ asymptotically in $m$ \ed{if the limit exists}, i.e.,
	\begin{equation}
		\epsilon^\star = \sup_{\epsilon_{\mathrm{ch}}\in[0,1)} \{\epsilon_{\mathrm{ch}}:\lim_{m\rightarrow\infty} P_{\SC}\left(\code^{\ed{[m]}}\right) = 0\}.
	\end{equation}
\end{definition}
As it is not possible to evaluate $P_{\SC}\left(\code^{\ed{[m]}}\right)$ exactly, we rely on the upper bound \eqref{eq:LooseUB_Product} to obtain a lower bound $\epsilon^\star_{\mathrm{LB}}$ on the block erasure threshold in the form
\begin{equation}
	\epsilon^\star_{\mathrm{LB}} = \hspace{-1mm}\sup_{\epsilon_{\mathrm{ch}}\in[0,1)} \hspace{-1mm}\{\epsilon_{\mathrm{ch}}:\lim_{m\rightarrow\infty} k\left(\code^{\ed{[m]}}\right)\epsilon_{\mathrm{max}}\left(\code^{\ed{[m]}}\right)\hspace{-1mm}=\hspace{-1mm}0\}
\end{equation}
where the dependence of the code dimension $k$ and of the maximum information bit erasure probability $\epsilon_{\mathrm{max}}$ on the sequence of product codes has been made explicit. We provide next two examples of product code sequences, whose rates converge to a positive value. The first sequence exhibits a positive block erasure threshold (lower bound), which is however arguably far from the Shannon limit. We then analyze a product code sequence that is known to achieve the \ac{BEC} capacity under bit-wise \ac{MAP} decoding \cite{kumar16,pfister17}.

\medskip
\begin{example}[Euler's infinite-product representation of the sine function as an \ac{SPC} product code]
\label{sec:asymptotic1}
Consider an \ac{SPC} product code \GL{sequence} with an $(a^2 \ell^2, a^2 \ell^2-1)$ \GL{\ac{SPC} component code at the $\ell$-th dimension, yielding $R_\ell = \left(1- (a\ell)^{-2}\right)$, with $\ell=1,\ldots,m$.} \GL{The asymptotic} rate is computed easily \mcc{via} the Euler's infinite-product representation of the sine \mcc{function}, i.e.,
\[\sin\left(\frac{\pi}{a}\right)=\frac{\pi}{a} \prod_{\ell=1}^{\infty}\left(1-\frac{1}{a^2\ell^2}\right)\] 
yielding an \mc{asymptotic} rate $R = \frac{a}{\pi}\sin\left(\frac{\pi}{a}\right)$. \GL{Different product code sequences can be obtained for various choices of the parameter $a$.}
\last{The} lower bounds \last{on the} block erasure thresholds are provided in Table \ref{table:asymptotic} \last{for several values of $a$.}
The second column \GL{in Table \ref{table:asymptotic} provides the asymptotic rate of the \ac{SPC} product code sequence defined by the parameter $a$ (whose squared value is reported in the first column). The third column reports the lower bound on the block erasure threshold. The fourth column gives the \mc{Shannon limit} for the given asymptotic rate, while the last one shows \mc{the fraction of the Shannon limit achieved by each construction.} The thresholds achieved by the different product code sequences lie relatively far from the Shannon limit. In relative terms, the lowest-rate construction (obtained for $a^2=2$) achieves the largest fraction (above $1/2$) of the limit, while the efficiency of the sequences decreases as the rate grows.}
{\renewcommand{\arraystretch}{1.2}
	\begin{table}
		\caption{Lower bounds on the \GL{block} erasure thresholds for some \ac{SPC} product code  \GL{sequences based on Euler's infinite-product representation of the sine function}}
		\begin{center}
			\vspace*{-3mm}
			\begin{tabular}{c|cc|c|c}
				\hline\hline
				$a^2$             & $R$ & $\epsilon^\star_{\mcc{\mathrm{LB}}}$ & Limit, $\epsilon=1-R$ & $\epsilon^\star_{\mcc{\mathrm{LB}}} /\epsilon$\\
				\hline
				$2$ & $0.3582$ & $0.3308$ & $0.6418$ & $0.5154$\\
				$4$ & $0.6366$ & $0.1440$ & $0.3634$ & $0.3963$\\
				$8$ & $0.8067$ & $0.0681$ & $0.1933$ & $0.3523$ \\
				$16$ & $0.9003$ & $0.0332$ & $0.0997$ & $0.3331$ \\
				$32$ & $0.9494$ & $0.0164$ & $0.0506$ & $0.3241$ \\
				$64$ & $0.9745$ & $0.0081$ & $0.0255$ & $0.3176$ \\ \hline\hline
			\end{tabular}
		\end{center}
		\label{table:asymptotic}
	\end{table}
}
\end{example}

\medskip

\begin{example}[Product of $(m,m-1)$ \ac{SPC} product codes in $m$ dimensions]
    \label{sec:asymptotic2}
We consider now the product code obtained by iterating $(m,m-1)$ \ac{SPC} codes in $m$ dimensions\ed{, i.e., the resulting code is an $(m^m,(m-1)^m,2^m)$ code. Hence, the} rate of the $m$th product code in the sequence is 
\begin{equation*}
	R\left(\code^{\ed{[m]}}\right) = \left(1-\frac{1}{m}\right)^{m}
\end{equation*}
and it converges, for $m\rightarrow \infty$, to $\texttt{e}^{-1}$.
As presented in \cite{pfister17}, this product code sequence is remarkable: It is capacity-achieving over the \GL{BEC} under bit-wise \ac{MAP} decoding. \ed{This} observation follows from results derived in \cite{kumar16}. Disappointingly, the block-wise erasure threshold under \ac{SC} decoding turns out to be zero. This (negative) result is provided by the following theorem.

\begin{theorem}
	\label{theorem:threshold}
	\GL{Under \ac{SC} decoding, the block erasure threshold of the product code sequence defined by the component code sets $\mathscr{C}^{\ed{[m]}}=\left\{\code_1, \code_2, \ldots, \code_m\right\}$, where $\code_i$, $i=1,\dots,m$, are $(m,m-1)$ \ac{SPC} codes, is zero.}
\end{theorem}
\begin{proof}
		The proof revolves around the following idea: For the product code with $(m,m-1)$ \ac{SPC} component codes in $m$ dimensions over the \ac{BEC}, the largest information bit erasure probability $\epsilon_{\mathrm{max}}$ under \ac{SC} decoding is equal to the channel erasure probability $\epsilon_{\mathrm{ch}}\in(0,1)$, i.e., $\epsilon_{\mathrm{max}} = \epsilon_{\mathrm{ch}}$ as $m\rightarrow\infty$\ed{, which is proved below after this paragraph}. Since the largest information bit erasure probability $\mcc{\epsilon}_{\mathrm{max}}$ is a lower bound on the block error probability under \ac{SC} decoding (see \eqref{eq:UB_Product}), we have that $\epsilon^\star = 0$.
		
		\begin{figure}[t]
			\begin{center}
				\includegraphics[width=\columnwidth]{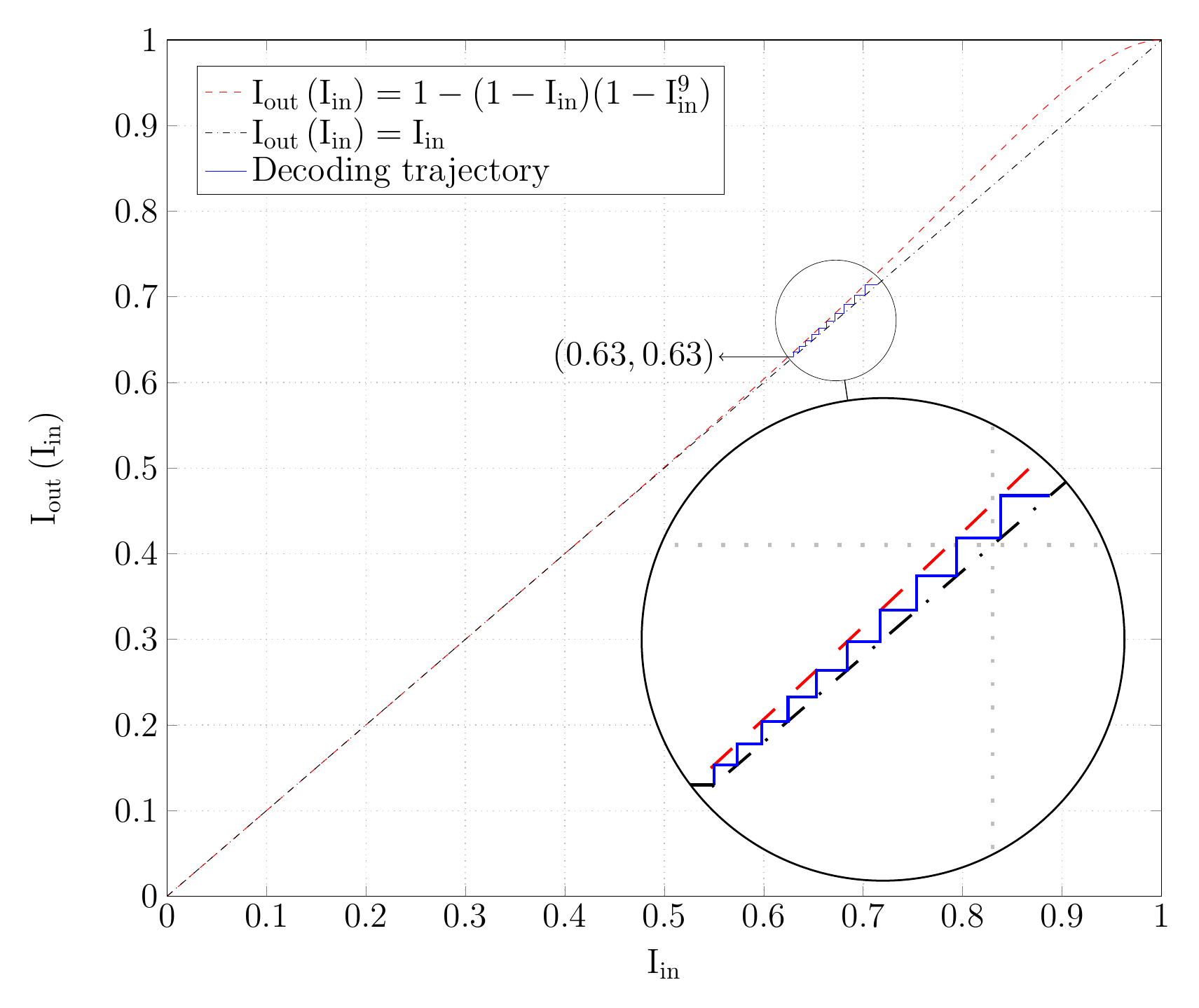}
			\end{center}
			\vspace{-0.5cm}
			\caption{Decoding trajectory for the bit $u_1$ of the $10$-dimensional SPC product code, where the component codes are $(10,9)$ SPC codes, over the BEC$\left(0.37\right)$.}\label{fig:decoding_trajectory}
		\end{figure}
		We rewrite the recursion \eqref{eq:erasure_SPC_PC}, for $i=1$ and $j=1$, in mutual information as
		\begin{equation}
		    \ed{I^{(1)}_{\G^{[m]}}}=1-\left[1-\ed{I^{(1)}_{\G^{[m-1]}}}\right]\left[1-\ed{\left(I^{(1)}_{\G^{[m-1]}}\right)}^{n_m-1}\right]
		\label{eq:mi_SPC_PC}
		\end{equation}
		by \ed{noting that} $\ed{I^{(i)}_{\G^{[m]}}} = 1-\ed{\epsilon^{(i)}_{\G^{[m]}}}$ \ed{where $I^{(i)}_{\G^{[m]}}$ denotes the mutual information of the \ac{BEC} with an erasure probability of $\epsilon^{(i)}_{\G^{[m]}}$ with uniform inputs}.
		We are interested in $\mi_\mathrm{min} \triangleq 1-\epsilon_\mathrm{max}$ for a given \ac{BEC}$(\epsilon_\mathrm{ch})$, with $\epsilon_\mathrm{ch}\in[0,1)$, which can be calculated recursively via \eqref{eq:mi_SPC_PC}. This recursion is illustrated \last{in Fig. \ref{fig:decoding_trajectory}} for the case \last{$m=10$ and $\epsilon_{\mathrm{ch}}=0.37$} as an example. Note that, for the considered construction, we have $n_\ell = m$, $\ell = 1,\dots,m$. \mc{This means \last{that} the top curve in the figure shifts down for a larger $m$, resulting in a narrower tunnel between the two curves, although the number of recursions, equal to $m$, increases. Note that it is necessary to have $\mi_\mathrm{min}\rightarrow 1$, i.e., one reaches \ed{the} $(1,1)$ point, with $m$ recursions in the figure for $P_{\SC}\rightarrow 0$. In the following, we provide an answer for the question on the dominating effect (narrower tunnel or more recursions) with increasing $m$.}
		
		The first (and single) recursion is simplified, using \ed{the \ac{RHS} of \eqref{eq:mi_SPC_PC} by setting the input} $\mi_{\mathrm{ch}}\triangleq 1-\epsilon_{\mathrm{ch}}$, as the following input-output relationship
		\begin{align}
		    \ed{f}(\mi_\mathrm{ch}) &\triangleq 1-\left(1-\mi_\mathrm{ch}\right)\left(1-\mi_\mathrm{ch}^{m-1}\right)\nonumber\\
		    &=\mi_{\mathrm{ch}}+\mi_{\mathrm{ch}}^{m-1}-\mi_{\mathrm{ch}}^m.
		    \label{eq:mi_evol}
		\end{align}
		Hence, we know that the mutual information after $m$ iterations is	$\mi_\mathrm{min} = \ed{f}^{\circ m}\left(\mi_\mathrm{ch}\right)$ where $\ed{f}^{\circ m}\left(\mi_\mathrm{ch}\right) = \ed{f}\left(\ed{f}^{\circ m-1}\left(\mi_\mathrm{ch}\right)\right)$ denotes the $m$-th iteration of function $\ed{f}$ with $\ed{f}^{\circ 1}\left(\mi_\mathrm{ch}\right)\triangleq \ed{f}\left(\mi_\mathrm{ch}\right)$. We are interested in the lowest channel mutual information for which the block error rate converges to zero asymptotically in $m$, i.e.,
		\[\mi^\star = \inf_{\mi_\mathrm{ch}\in(0,1]}\{\mi_\mathrm{ch}:\lim_{m\rightarrow\infty} P_{\SC} \rightarrow 0 \}.\]
		Consider now \ed{an arbitrary $\delta>0$. For any} positive $\gamma<1$, there exists a sufficiently large $m$ such that $m\gamma^{m-2}\leq\delta$. Then, for any \ed{non-negative $\mi\leq\gamma$, we write
		\begin{align}
		    \ed{f}(\mi) &= \mi+\mi^{m-1}-\mi^{m} \label{eq:upper_bound_mi1}\\
		    &\leq \mi\left(1+\mi^{m-2}\right)\label{eq:upper_bound_mi2} \\
		    &\leq \mi\left(1+\frac{\delta}{m}\right)
		\label{eq:upper_bound_mi}
		\end{align}
		where \eqref{eq:upper_bound_mi2} follows from the fact that $\mi\geq0$ and \eqref{eq:upper_bound_mi} from the fact that $\mi\leq\gamma$, which, combined with $m\gamma^{m-2}\leq\delta$, leads to $\mi^{m-2}\leq\frac{\delta}{m}$.} For any \ed{initial} $\mi_{\mathrm{ch}}\leq\texttt{e}^{-\delta}\gamma$ and any $m'\leq m$, we have
		\begin{equation}
		    \ed{f}^{\circ m'}(\mi_\mathrm{ch}) \leq \texttt{e}^{-\delta}\gamma \left(1+\frac{\delta}{m}\right)^{m'} \leq\gamma
		\end{equation}
		which makes sure that the condition \ed{$\mi\leq\gamma$} for \eqref{eq:upper_bound_mi} is not violated with $m'$ iterations. Therefore, for any positive $\gamma < 1$, any $\delta>0$ and any $\mi_{\mathrm{ch}}\leq\texttt{e}^{-\delta}\gamma$, we write
		\begin{equation}
		    \lim_{m\rightarrow\infty} \ed{f}^{\circ m}(\mi_\mathrm{ch}) \leq \mi_\mathrm{ch}\texttt{e}^{\delta}.
		\end{equation}
		The result follows by choosing $\delta$ arbitrarily small \ed{and $\gamma$ arbitrarily close to $1$}.
\end{proof}
\end{example}

\subsection{Analysis over Binary Memoryless Symmetric Channels}
\label{sec:DE}

\ed{It is well-known that the block error probability of polar codes under \ac{SC} decoding can be analyzed using density evolution, when the transmission takes place over a \ac{BMS} channel\cite{Mori:2009SIT}. In the following, the same method is used to provide a tight upper bound on the block error probability of \ac{SPC} product codes over \ac{BMS} channels.}

\ed{Due to the channel symmetry and the linearity of the codes, we assume that the all-zero codeword is transmitted. We write $L^{(i)}_{\ed{\G^{[m]}}}(\vecy,\boldsymbol{0})$ to denote the the log-likelihood ratio for $u_i$ based on \eqref{eq:recurs_metric} where all the previous bit-values are provided as zeros to the decoder, i.e.,
\begin{equation*}
	L^{(i)}_{\ed{\G^{[m]}}}(\vecy,\boldsymbol{0}) \triangleq \log\frac{\W^{(i)}_{\ed{\G^{[m]}}}(\vecy,U_1^{i-1} = \boldsymbol{0}|U_i=0)}{\W^{(i)}_{\ed{\G^{[m]}}}(\vecy,U_1^{i-1} =\boldsymbol{0}|U_i=1)}.
\end{equation*}
Accordingly, we use $l^{(i)}_{\ed{\G^{[m]}}}$ to denote the \ac{p.d.f.} of the \ac{RV} $L^{(i)}_{\ed{\G^{[m]}}}(Y_1^n,\boldsymbol{0})$. Extending the equations \eqref{eq:erasure_SPC} and \eqref{eq:erasure_SPC_PC} to general \ac{BMS} channels, the densities can be computed recursively as
\begin{equation}
	l^{(i)}_{\ed{\G^{[m]}}} = l^{(j +1)}_{\ed{\G^{[m-1]}}} \varoast (l^{(j +1)}_{\ed{\G^{[m-1]}}})^{\boxast (n_m-t)}
	\label{eq:de_awgn}
\end{equation}
\own{with $j = \lfloor\nicefrac{(i-1)}{k_m}\rfloor$ and $ t \triangleq [(i-1)\!\!\mod k_m] +1$} where  $\varoast$ denotes the variable node convolution and $(l^{(j +1)}_{\ed{\G^{[m-1]}}})^{\boxast (a)}$  the $a$-fold check node convolution with $(l^{(j +1)}_{\ed{\G^{[m-1]}}})^{\boxast (1)} \triangleq l^{(j +1)}_{\ed{\G^{[m-1]}}}$ as defined in \cite[Ch. 4]{Richardson:2008:MCT:1795974}. Then, the \ac{RHS} of \eqref{eq:SC_bound} can be computed via $l^{(i)}_{\ed{\G^{[m]}}}$ with $i = 1,\dots,k$, i.e., as
\begin{equation}
    \sum_{i=1}^{k}\lim\limits_{z\rightarrow 0} \left(\int_{-\infty}^{-z}l^{(i)}_{\ed{\G^{[m]}}}(x)dx+\frac{1}{2}\int_{-z}^{+z}l^{(i)}_{\ed{\G^{[m]}}}(x)dx\right).\label{eq:UB_AWGN}
\end{equation}
The computation of \eqref{eq:de_awgn} and \eqref{eq:UB_AWGN} can be carried out, for instance, via quantized density evolution \cite{chung_urbanke2001}, yielding an accurate estimate of the \ac{RHS} of \eqref{eq:SC_bound}.}

\ed{In Fig. \ref{fig:sc_125_64_SPC}, we provide simulation results for the $(125,64)$ \ac{SPC} product code over the \ac{B-AWGN} channel. The results are provided in terms of \ac{BLER} vs. \ac{SNR}, where the \ac{SNR} is expressed as $E_b/N_0$ ($E_b$ is here the energy per information bit and $N_0$ the single-sided noise power spectral density). The \ac{SC} decoding performance is compared to the performance under Elias' and \ac{BP} decoding. While it was only proven for the case of \ac{BEC} in Theorem \ref{theorem:comparison}, the results illustrate that the \ac{SC} decoding outperforms Elias' decoding also over the \ac{B-AWGN} channel for the simulated code. \ac{BP} decoding with a maximum number of iterations set to $100$ outperforms the \ac{SC} decoding significantly, which motivates us to introduce \ac{SCL} decoding in the next section. The tight upper bound on the \ac{SC} decoding, computed via \eqref{eq:UB_AWGN}, is also provided.
\begin{figure}[t]
	\begin{center}
		\includegraphics[width=\columnwidth]{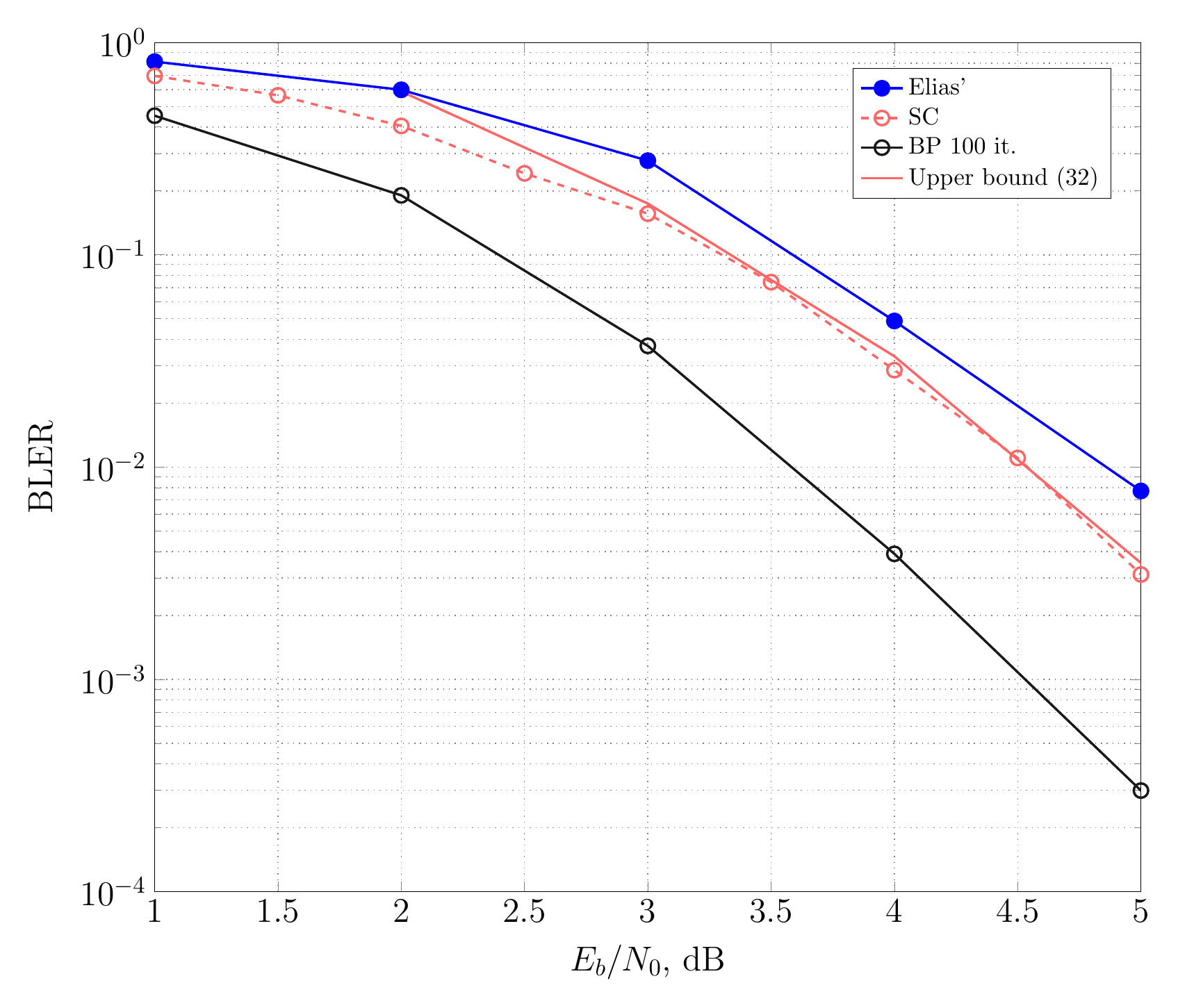}
	\end{center}
	\vspace{-0.5cm}
	\caption{\ed{\ac{BLER} vs. \ac{SNR} over the B-AWGN channel for a $(125,64)$ product code under \ac{SC} decoding, compared to Elias' decoding as well as a \ac{BP} decoding with $100$ iterations.}}\label{fig:sc_125_64_SPC}
\end{figure}}

%% file: sc_list.tex
\section{Successive Cancellation List Decoding of Single Parity-Check Product Codes}
\label{sec:scl}

While the asymptotic analysis provided in Section \ref{subsec:asymptotic} provides some insights on the \ac{SPC} product codes constructed over a large number of dimensions (yielding very large block lengths), we are ultimately interested in the performance of product codes in the practical setting where the number of dimensions is small, and the block length is moderate (or small). \last{Like} polar codes, \ac{SC} decoding of \ac{SPC} product codes \ed{performs} poorly in this regime\ed{, e.g., see Fig. \ref{fig:sc_125_64_SPC}}. Hence, following the footsteps of \cite{tal15}, \last{we} investigate the error probability of \ac{SPC} product codes under \ac{SCL} decoding.
\mcc{The} \ac{SC} decoder decides on the value of $u_i$, $1\leq i\leq k$, after computing the corresponding likelihood. The \ac{SCL} decoder \last{does not make} a final decision on the bit value $u_i$ immediately. Instead, it \last{considers both options (i.e., $\hat{u}_i = 0$ and $\hat{u}_i = 1$) and runs several instances of \last{the} \ac{SC} decoder in parallel}. Each \last{\ac{SC} decoder} corresponds to a decoding \emph{path}, \mc{defined by a given} different hypothesis on $i-1$ preceding \mc{information} bits, namely $\hat{u}_1^{i-1}$, $1 < i < k$, at decoding stage $i$. After computing the likelihood for $u_i$, \mc{a} path is split into two new paths, \last{that share the} same decisions $\hat{u}_1^{i-1}$ for the preceding $i-1$ bits. \last{But these} new paths consider \last{different} decisions \last{for $\hat{u}_i$ (i.e., $\hat{u}_i = 0$ and $\hat{u}_i = 1$),} respectively, which doubles the number of paths at each \mc{decoding stage}. \mc{In order to avoid an \GL{excessive} growth in the number of paths, a maximum list size $L$ is imposed. The} \ac{SCL} decoder discards the paths \mc{except for the most likely paths, according to metrics computed via \eqref{eq:recurs_metric},} whenever their number exceeds \mc{the} list size $L$.
After $k$ steps, among the surviving $L$ candidates, the \ac{SCL} decoder chooses as the final decision the candidate \mc{path} maximizing the likelihood\mc{, yielding a decision $\vecuhat$}. Obviously, the block error rate decreases with an increasing $L$ at the expense of a higher complexity.
\begin{remark}\label{remark:ML}
	For \ed{short- and moderate-length} polar codes, it was shown in \cite{tal15} that \last{close-to-\ac{ML}} performance can be attained with a sufficiently large \ed{(yet, manageable)} list size. This was demonstrated by computing a numerical lower bound on the \ac{ML} decoding error probability via Monte Carlo simulation, where the correct codeword is introduced artificially in the final list, prior to the final \mcc{decision}. If, for a specific list size $L$, the simulated error probability is close to the numerical \ac{ML} decoding lower bound, then increasing the list size $L$ would not yield any performance improvement. The same principle applies to \ac{SCL} decoding of \ac{SPC} product codes.
\end{remark}
Owing to this observation, we \ed{first} study the performance of \ac{SPC} product codes under \ac{ML} decoding, by developing a weight enumerator analysis. Similarly, inspired by the concatenated polar code construction of \cite{tal15}, we \ed{also} study the performance of a concatenation of a high-rate outer code with an inner \ac{SPC} product code under \ac{ML} decoding. \ed{For both cases, t}he analysis is complemented by the \ac{SCL} decoding simulations\ed{. We will see that (for some short product codes) the \ac{ML} decoding performance is indeed attainable by \ac{SCL} decoding with small list sizes, e.g., $L\leq 8$, while moderate list sizes, e.g., $128\leq L\leq 1024$, are required when \ac{SPC} product codes are concatenated with an outer code}.

\ed{\subsection{Finite-length Performance Analysis via Weight Distribution of SPC Product Codes}
\label{sec:wef_SPC_PC}
Computing the weight enumerator of \ac{SPC} product codes for small constructions is feasible using the method presented in \cite{caire}. First, we provide an alternative derivation to the \ac{WEF} of a $2$-dimensional product code $\code$ with systematic arbitrary binary linear component code $\code_1$ and systematic $(\nu,\nu-1)$ \ac{SPC} code $\SPC_\nu$ as the second component code.
\begin{theorem}
	\label{theorem:wef_spc}
	Let $\code_1$ and $\SPC_\nu$ be an arbitrary $(n_1,k_1)$ systematic code with a generator matrix \ed{$\boldsymbol{G}$} and a length-$\nu$ systematic \ac{SPC} code, respectively. Then, the \ac{WEF} of the product code $\code$ with component codes $\code_1$ and $\SPC_\nu$ is
	\begin{equation}
		A_{\code}(z) = 2^{-k_{\scaleto{1}{3pt}}}\sum_{\vecv\in\{0,1\}^{k_{\scaleto{1}{3pt}}}} \left(\sum_{\vecu\in\{0,1\}^{k_{\scaleto{1}{3pt}}}}(-1)^{\vecu\cdot\vecv^\mathrm{T}}z^{w_{\scaleto{\mathrm{H}}{3pt}}\left(\vecu\ed{\boldsymbol{G}}\right)}\right)^\nu.\label{eq:theorem_wef}
	\end{equation}
\end{theorem}}
\begin{proof}
    \ed{Let $f(\boldsymbol{Z}) = \prod_{i=1}^\nu f_i(z_{i,1},\ldots,z_{i,n})$ be a multinomial in the variables $\boldsymbol{Z} = \{z_{i,j}\}$, $1\leq i\leq \nu$, $1\leq j\leq \own{n}$, where each factor $f_i$ is a multinomial only in the variables $z_{i,1},\ldots,z_{i,n}$. Assume further that each variable $z_{i,j}$ appears with exponent either $0$ or $1$ in the multinomial $f_i(z_{i,1},\ldots,z_{i,n})$, and, hence, in $f(\boldsymbol{Z})$.}
	\ed{Suppose now that we wish to remove from $f(\boldsymbol{Z})$ all the terms in the form \begin{equation}
	    \prod_{i\in\mathcal{T}}z_{i,j}
	\end{equation}
	where $|\mathcal{T}|$ is odd. The remaining terms can be obtained by computing
	\begin{equation}
	    \frac{1}{2}\sum_{y\in\{+1,-1\}}\prod_{i=1}^\nu f_i(z_{i,1},\ldots,z_{i,j-1},yz_{i,j},z_{i,j+1},\ldots,z_{i,n}).
	\end{equation}
	Similarly, to remove all the terms in the form
	\begin{equation}
	    \prod_{i\in\mathcal{T}_j}z_{i,j},\quad j\in[n]
	\end{equation}
	where $|\mathcal{T}_j|$ is odd, it is sufficient to evaluate
	\begin{align}
	    2^{-n}\sum_{\vecy\in\{+1,-1\}^n}\prod_{i=1}^\nu f_i(y_1z_{i,1},\ldots,y_nz_{i,n}).
	\end{align}
     }
	
	\ed{Consider first a product code composed of $\nu\times n_1$ arrays whose rows and columns are codewords of $\code_1$ and a trivial rate-$1$ code $\mathcal{I}$ with a generator matrix $\I_\nu$, respectively. Then, the complete \ac{WEF} $A_{\code_1\otimes\mathcal{I}}(\boldsymbol{Z})$ of the product code uses the dummy variables $\boldsymbol{Z} = \{z_{i,j}\}$, $1\leq i\leq \nu$, $1\leq j\leq n_1$, to track bits by their $(i,j)$ coordinate in the codeword. This is obtained simply by the multiplication of the complete \acp{WEF} for the codes corresponding to each row $i$, i.e., we have
	\begin{equation}
	\label{eq:comp_wef_prod_trivial}
		A_{\code_1\otimes\mathcal{I}}(\boldsymbol{Z}) = \prod_{i=1}^{\nu}A_{\code_{\scaleto{1}{3pt}}}(z_{i,1},\ldots,z_{i,n_{\scaleto{1}{3pt}}}).
	\end{equation}
	Recall now that the codewords of $\code$ are $\nu\times n_1$ arrays whose rows and columns are codewords of $\code_1$ and $\SPC_\nu$, respectively. Then, the complete \ac{WEF} $A_{\code}(\boldsymbol{Z})$ of the product code $\code$ is derived from $A_{\code_1\otimes\mathcal{I}}(\boldsymbol{Z})$ by imposing that each column word has an even weight, i.e., we have
	\begin{align}
		\!\!\!\!\!\!\!\!A_{\code}(\boldsymbol{Z}) = 2^{-n_{\scaleto{1}{3pt}}}\sum_{\vecy\in\{+1,-1\}^{n_{\scaleto{1}{3pt}}}} \prod_{i=1}^{\nu}A_{\code_{\scaleto{1}{3pt}}}(y_1z_{i,1},\dots,y_{n_{\scaleto{1}{3pt}}}z_{i,n_{\scaleto{1}{3pt}}}).\label{eq:comp_wef_prod}
	\end{align}
	Then, the \ac{WEF} of the product code $\code$ is obtained by setting $z_{i,j} = z$, yielding
	\begin{align}
        &A_{\code}(z) = 2^{-n_{\scaleto{1}{3pt}}}\sum_{\vecy\in\{1,-1\}^{n_{\scaleto{1}{3pt}}}} \left(A_{\code_{\scaleto{1}{3pt}}}(y_1z,\dots,y_{n_{\scaleto{1}{3pt}}}z)\right)^\nu \label{eq:wef_step1}\\
        &= 2^{-n_{\scaleto{1}{3pt}}}\sum_{\vecy\in\{1,-1\}^{n_{\scaleto{1}{3pt}}}}\left(\sum_{\vecx\in\code_1} z^{w_{\scaleto{\mathrm{H}}{3pt}}(\vecx)}\prod_{i^\prime = 1}^{n_1}y_{i^\prime}^{x_{i^\prime}}\right)^{\nu} \label{eq:wef_step2}\\
        &= 2^{-n_{\scaleto{1}{3pt}}}\sum_{\vecy\in\{0,1\}^{n_{\scaleto{1}{3pt}}}}\left(\sum_{\vecx\in\code_1} z^{w_{\scaleto{\mathrm{H}}{3pt}}\left(\vecx\right)}(-1)^{\vecx\cdot\vecy^{\mathrm{T}}} \right)^\nu \label{eq:wef_step4}\\
        &= 2^{-n_{\scaleto{1}{3pt}}}\sum_{\vecy\in\{0,1\}^{n_{\scaleto{1}{3pt}}}}\left(\sum_{\vecu\in\{0,1\}^{k_{\scaleto{1}{3pt}}}}z^{w_{\scaleto{\mathrm{H}}{3pt}}\left(\vecu\G\right)}(-1)^{\vecu\G\vecy^{\mathrm{T}}} \right)^{\nu} \label{eq:wef_step_rev1}\\
        &= 2^{-n_{\scaleto{1}{3pt}}}\sum_{\vecy\in\{0,1\}^{n_{\scaleto{1}{3pt}}}}\left(\sum_{\vecu\in\{0,1\}^{k_{\scaleto{1}{3pt}}}}z^{w_{\scaleto{\mathrm{H}}{3pt}}\left(\vecu\G\right)}(-1)^{\vecu\G\tilde{\boldsymbol{H}}\vecy^{\mathrm{T}}} \right)^{\nu} \label{eq:wef_step5}\\
        &= 2^{-n_{\scaleto{1}{3pt}}}\!\sum_{y_1^{n_{\scaleto{1}{3pt}}-k_{\scaleto{1}{3pt}}}}\!\sum_{\vecv\in\{0,1\}^{k_{\scaleto{1}{3pt}}}}\!\!\left(\sum_{\vecu\in\{0,1\}^{k_{\scaleto{1}{3pt}}}}\!\!z^{w_{\scaleto{\mathrm{H}}{3pt}}\left(\vecu\G\right)}(-1)^{\vecu\G\bar{\boldsymbol{H}}^{\mathrm{T}}\vecv^{\mathrm{T}}}\!\right)^{\nu} \label{eq:wef_step6}
    \end{align}
with $\tilde{\boldsymbol{H}}\triangleq\left[\boldsymbol{H}^{\mathrm{T}},\bar{\boldsymbol{H}}^{\mathrm{T}}\right]$ where $\boldsymbol{H}$ and $\bar{\boldsymbol{H}}$ are the parity-check matrix of $\code_1$ and a complementary matrix such that $\tilde{\boldsymbol{H}}$ is non-singular, respectively, and $\vecv\triangleq y_{n_{\scaleto{1}{3pt}}-k_{\scaleto{1}{3pt}}+1}^{n_{\scaleto{1}{3pt}}}$. Note that \eqref{eq:wef_step2} follows from \eqref{eq:cwef} by noting that $(y_1z,\ldots,y_{n_{\scaleto{1}{3pt}}}z)^{\vecx} = z^{w_{\scaleto{\mathrm{H}}{3pt}}(\vecx)}\vecy^{\vecx}$, \eqref{eq:wef_step4} from re-defining the dummy vector $\vecy$, \ed{\eqref{eq:wef_step_rev1} from $\vecx=\vecu\boldsymbol{G}$ and performing the summation over $\vecu\in\{0,1\}^{k_{\scaleto{1}{3pt}}}$ instead of $\vecx\in\code_{\scaleto{1}{3pt}}$, \eqref{eq:wef_step5} from the nonsingularity of $\tilde{\boldsymbol{H}}$ and the summation being over all possible $\vecy\in\{0,1\}^{n_{\scaleto{1}{3pt}}}$}, and \eqref{eq:wef_step6} from dividing the outer summation into two parts, namely over $y_1^{n_{\scaleto{1}{3pt}}-k_{\scaleto{1}{3pt}}}\in\{0,1\}^{n_{\scaleto{1}{3pt}}-k_{\scaleto{1}{3pt}}}$ and $\vecv\in\{0,1\}^{k_{\scaleto{1}{3pt}}}$, and the fact that the product $\G\boldsymbol{H}^{\mathrm{T}}$ results in the all-zero matrix. Finally, \eqref{eq:theorem_wef} follows from the fact that the outer summation can be removed by multiplying the remaining term by $2^{n_{\scaleto{1}{3pt}}-k_{\scaleto{1}{3pt}}}$ and the fact that the product $\G\bar{\boldsymbol{H}}^{\mathrm{T}}$ is non-singular.\footnote{\ed{Note that the rows of $\bar{\boldsymbol{H}}^{\mathrm{T}}$ are linearly independent of the rows of $\boldsymbol{H}$ by definition.}}}
\end{proof}
\ed{Thanks to Theorem~\ref{theorem:wef_spc}, one can compute the \ac{WEF} of short and moderate-length \ac{SPC} product codes iteratively, by simply choosing one component code to be a \ac{SPC} product code $\code_1$ and the other one to be a \ac{SPC} code $\SPC_\nu$. Given the weight enumerator of a product code, upper bounds on the \ac{ML} decoding error probability can be obtained. As an example, a tight bound on the block error probability over the \ac{B-AWGN} channel is provided by Poltyrev's \ac{TSB} \cite{poltyrev94}. Another example of a tight bound on the block error probability of a code based on its weight enumerator is Di's union bound over the \ac{BEC} \cite[Lemma B.2]{di02}. We refer the interested reader to \cite{sason2006performance} for an extensive survey on performance bounds under \ac{ML} decoding.}

\ed{Next, we provide simulation results for the $3$-dimensional $(125,64)$ \ac{SPC} product code under \ac{ML} decoding implemented  over the \ac{BEC} via \ac{SCL} decoding with $L=\infty$. The results are provided in Fig. \ref{fig:ml_125_64_SPC_vs_polar_bec}. As reference, simulation results for a $(125,64)$ punctured polar code~\cite{WR14:polar_punc} are also provided, where the polar code design follows the guidelines of the 5G standard~\cite{BCL21:5g_polar}. The polar code slightly outperforms the considered \ac{SPC} product code. Note that Di's union bound tightly approaches the performance of the \ac{SPC} product code.}
\begin{figure}[t]
	\begin{center}
		\includegraphics[width=\columnwidth]{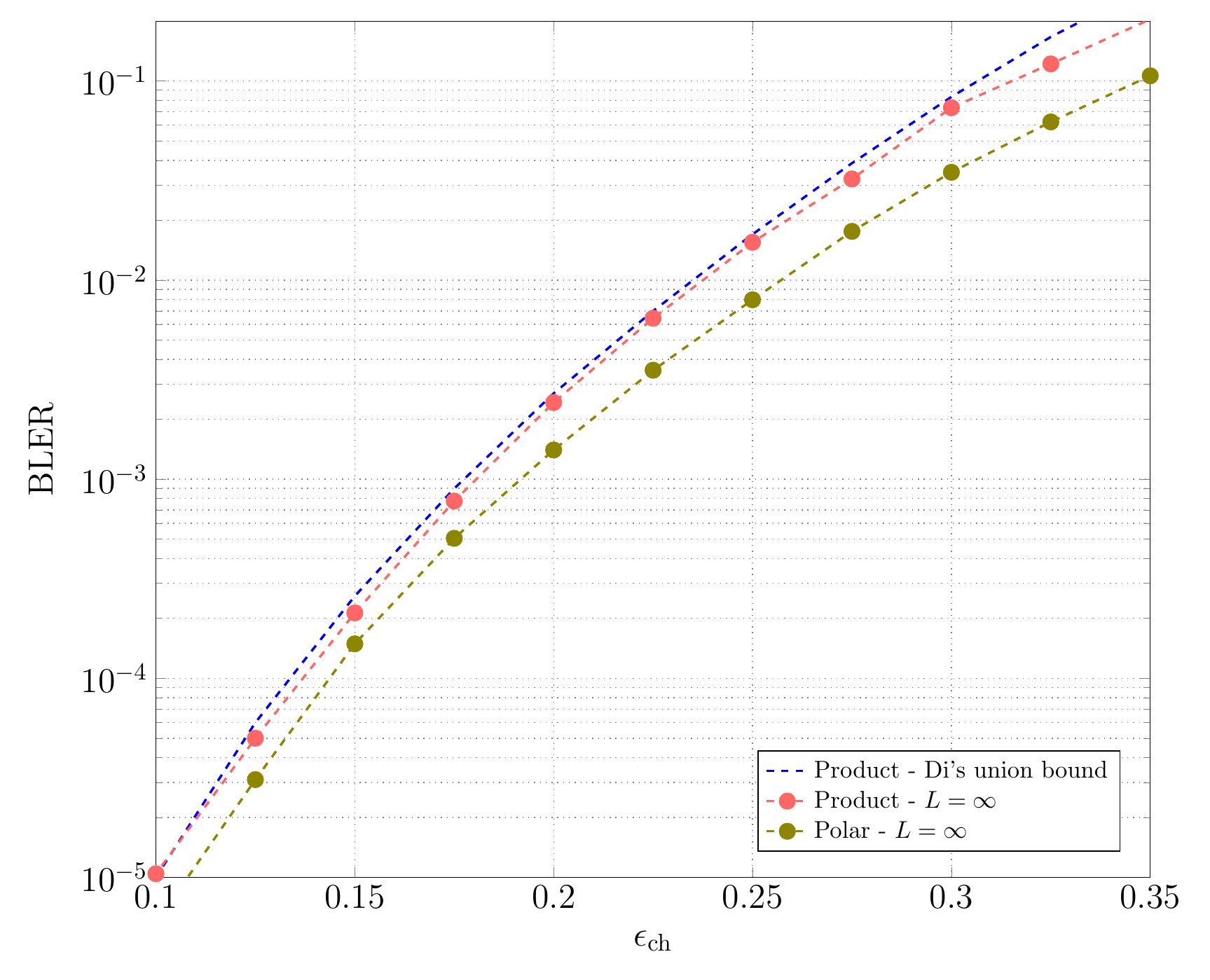}
	\end{center}
	\vspace{-0.5cm}
	\caption{\ed{\ac{BLER} vs. $\epsilon_{\mathrm{ch}}$ over the \ac{BEC} for the $(125,64)$ product code under \ac{ML} decoding, implemented via \ac{SCL} decoding with $L=\infty$, compared to a $(125,64)$ punctured polar code.}
	}\label{fig:ml_125_64_SPC_vs_polar_bec}
\end{figure}

\ed{The same $(125,64)$ \ac{SPC} product code is also simulated over the \ac{B-AWGN} channel under \ac{SCL} decoding with various list sizes. The results are given in Fig. \ref{fig:scl_125_64_SPC}. The \ac{TSB} is also provided, thanks to the weight enumerator analysis.} Remarkably, \ed{\ac{SCL}} decoding with $L=4$ is sufficient to operate \ed{very close to} the \ac{TSB} and \ed{to} outperform \ac{BP} decoding. With $L=8$, the \ac{SCL} decoder tightly \ed{approaches} the \ac{ML} lower bound, which is not the case for \ac{BP} decoding. \ed{The \ac{RCU}~\cite[Thm.~16]{Polyanskiy10:BOUNDS} and the \ac{MC}~\cite[Thm.~26]{Polyanskiy10:BOUNDS} bounds, are here plotted as reference.} The gap to the \ac{RCU} bound reaches to \ed{$1.7$} dB at \ed{\ac{BLER}} of $10^{-3}$.
\begin{figure}[t]
	\begin{center}
		\includegraphics[width=\columnwidth]{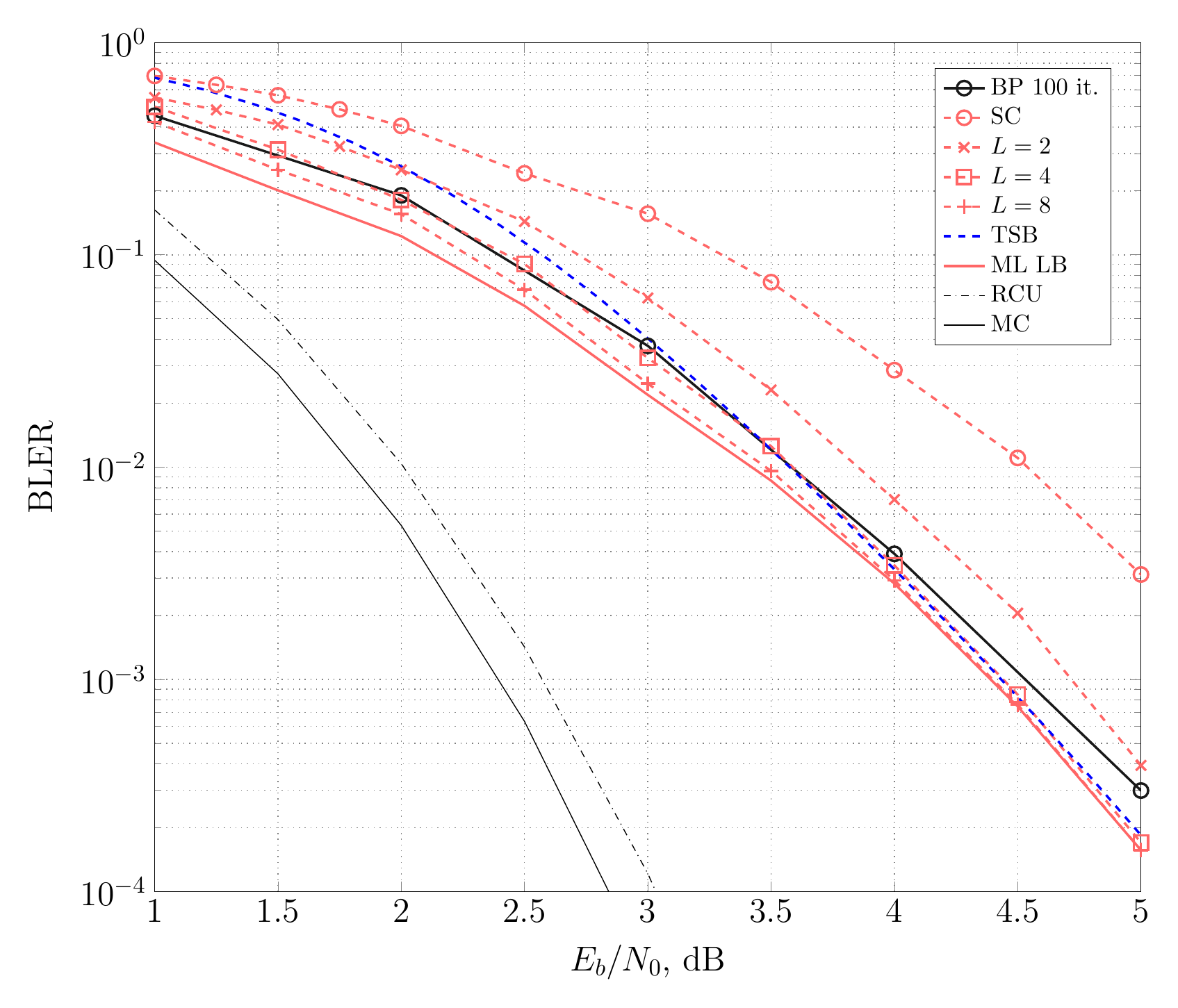}
	\end{center}
	\vspace{-0.5cm}
	\caption{\ed{\ac{BLER} vs. \ac{SNR}  under \ac{SCL} decoding for the $(125,64)$ product code with various list sizes, compared to a BP decoding with $100$ iterations.}
	}\label{fig:scl_125_64_SPC}
\end{figure}

\ed{The performance of \ac{SPC} product codes is compared to that of the $(125,64)$ polar code in Fig. \ref{fig:scl_125_64_SPC_vs_polar}. When $L=4$ is considered, the performance of polar code tightly matches its \ac{ML} lower bound and outperforms the \ac{SPC} product code by around $0.3$ dB at \ac{BLER} of $10^{-3}$. The gap between their \ac{ML} performance is around $0.25$ dB. Note that the polar code requires a smaller list size to approach its \ac{ML} performance.}
\begin{figure}[t]
	\begin{center}
		\includegraphics[width=\columnwidth]{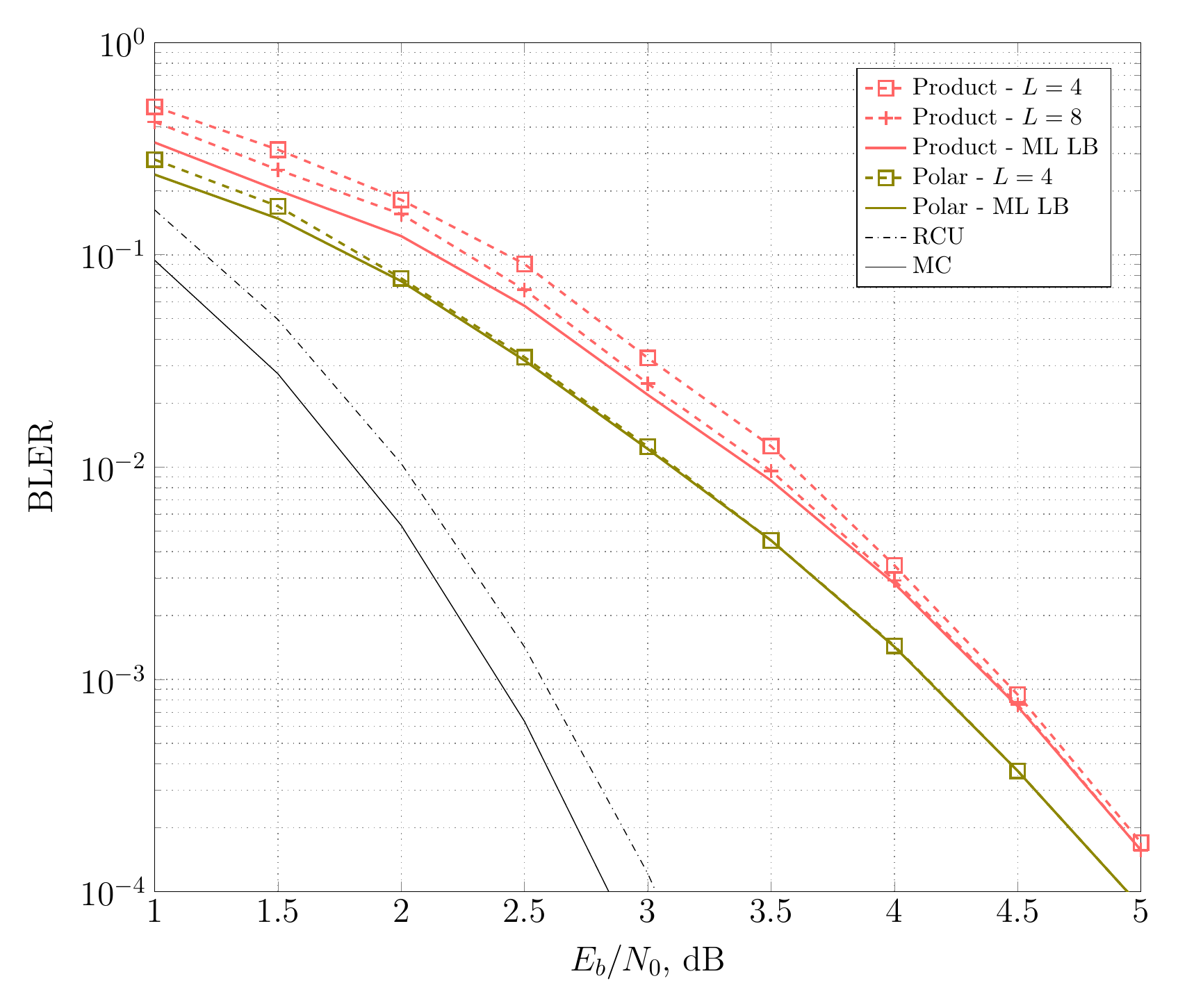}
	\end{center}
	\vspace{-0.5cm}
	\caption{\ed{\ac{BLER} vs. \ac{SNR} over the \ac{B-AWGN} channel for the $(125,64)$ product code under \ac{SCL} decoding with various list sizes, compared to a $(125,64)$ punctured polar code.}
	}\label{fig:scl_125_64_SPC_vs_polar}
\end{figure}

\subsection{Finite-length \GL{Performance} Analysis via Average Weight Distribution of Concatenated Ensembles}
\label{sec:concatenation}
Following \cite{tal15}, \ed{\ac{SPC} product codes are concatenated with a high-rate outer code to improve the distance profile, aiming at an improved performance under \ac{SCL} (and \ac{ML}) decoding. At the receiver, an \ac{SCL} decoder with list size $L$ is employed for the inner code.} The outer code is used to \GL{test the} $L$ codewords in the final list. Among those satisfying the outer code constraints, the most likely one is chosen as the final decision. \ed{Our goal is to analyze the \ac{ML} decoding performance of the concatenated code, since for short and moderate blocklengths close-to-ML performance may be practically attained via \ac{SCL} decoding by choosing a sufficiently large list. To analyze the \ac{ML} decoding performance of such a concatenation, we first \ed{derive} the weight enumerator of product codes in concatenation with an outer code. The weight enumerators will  then be used to derive well-known tight upper bounds on the \ac{ML} decoding error probability \ed{as before}. A reason to analyze such concatenation lies (besides in the expected performance improvement) in the fact that actual schemes employing product codes may make use of an error detection code to protect the product code information message (this is the case, for instance, of the IEEE~802.16 standard \cite{IEEE80216}). One may hence consider sacrificing (part of) the error detection capability introduced by the outer error detection code for a larger coding gain \cite{Coskun19:RM-Product}.}

Consider the concatenation of an $(n_{\mathsf i},k_{\mathsf i})$ inner product code $\code_{\mathsf{i}}$ with an $(n_{\mathsf o},k_{\mathsf o})$ high-rate outer code $\code_{\mathsf{o}}$ with $k_{\mathsf i}=n_{\mathsf o}$. The generator matrices of $\code_{\mathsf{i}}$ and $\code_{\mathsf{o}}$ are $\G_{\mathsf{i}}$ and $\G_{\mathsf{o}}$, respectively.
\begin{definition}[Concatenated Ensemble] The (serially) concateneted ensemble $\ensemble\left(\code_{\mathsf{o}},\code_{\mathsf{i}}\right)$ is the set  of all codes with generator matrix of the form
	\[
	\G=\G_{\mathsf{o}}\bm{\Pi}\G_{\mathsf{i}}
	\]
	where $\bm{\Pi}$ is an $n_{\mathsf o}\times n_{\mathsf o}$ permutation matrix.
\end{definition}
Assume that the outer code weight enumerator $A^{\mathsf o}_j$ and the input-output weight enumerator of the inner product code $A^{\mathrm{IO},\mathsf i}_{j,w}$ are known. We are interested in the weight enumerator of the concatenated code for a given permutation matrix $\bm{\Pi}$ (e.g., $\bm{\Pi} = \I_{n_{\mathsf o}}$), which interleaves the output of the outer encoder. Given an interleaver, this requires the enumeration of all possible input vectors of length $k_{\mathsf o}$, which is not possible in practice. For this reason, the impact of the outer code is \last{typically} analyzed in a concatenated ensemble setting by assuming that the interleaver is distributed uniformly over all possible ${n_\text{out}\choose i}$ permutations\cite{BenMon96}. Then, the average weight enumerator of the ensemble is derived as
\begin{equation}\label{eq:average_wef}
\bar{A}_w = \sum_{j = 0}^{n_\text{out}} \frac{A_j^{\mathsf o}\cdot A_{j,w}^{\mathrm{IO},\mathsf i}}{{n_{\mathsf o}\choose j}}
\end{equation}
where $\bar{A}_w$ is the average multiplicity of codewords $\ed{\vecx}$ with $w_{\Ham}(\ed{\vecx}) = w$. \ed{The computation of \eqref{eq:average_wef} requires the knowledge of the input-output weight enumerator of the inner \ac{SPC} product code. Therefore, we extend the result  of Theorem \ref{theorem:wef_spc} to the derivation of the \ac{IOWEF}.}

\begin{theorem}\label{theorem:iowef}
	\mc{Let $\code_1$ and $\SPC_\nu$ be an arbitrary $(n_1,k_1)$ systematic code with a generator matrix $\ed{\boldsymbol{G}}$ and a systematic \ac{SPC} code, respectively. Then, the \ac{IOWEF} of the product code $\code$ with component codes $\code_1$ and $\SPC_\nu$ is}
	\begin{align}
		A_{\code}^{\mathrm{IO}}(x,z) \!=\! 2^{-k_{\scaleto{1}{3pt}}}\!\sum_{\vecv\in\{0,1\}^{k_{\scaleto{1}{3pt}}}}\!\!&\left(\sum_{\vecu\in\{0,1\}^{k_{\scaleto{1}{3pt}}}}\!\!(-1)^{\vecu\cdot\vecv^\mathrm{T}}x^{w_{\scaleto{\mathrm{H}}{3pt}}\left(\vecu\right)}z^{w_{\scaleto{\mathrm{H}}{3pt}}\left(\vecu\ed{\boldsymbol{G}}\right)}\!\right)^{\nu-1}\\
		&\qquad\left(\sum_{\vecu\in\{0,1\}^{k_{\scaleto{1}{3pt}}}}(-1)^{\vecu\cdot\vecv^\mathrm{T}}z^{w_{\scaleto{\mathrm{H}}{3pt}}\left(\vecu\ed{\boldsymbol{G}}\right)}\right).
	\end{align}
\end{theorem}
\begin{proof}
	Recall \eqref{eq:comp_wef_prod} and assume, without loss of generality, that the component code generator matrices are of the form $[\I_{k_i} | \boldsymbol{P}_i]$. \ed{In this case, we set $z_{i,j} = xz$, for $1\leq i\leq k_1$ and $1\leq j\leq k_2 = \nu-1$, and $z_{i,j} = z$, otherwise, to obtain}
	\ed{\begin{align}
	    A_{\code}(x,z) = 2^{-n_{\scaleto{1}{3pt}}}\sum_{\vecy\in\{1,-1\}^{n_{\scaleto{1}{3pt}}}}&\left(\sum_{\vecu\in\{0,1\}^{k_{\scaleto{1}{3pt}}}}x^{w_{\scaleto{\mathrm{H}}{3pt}}\left(\vecu\right)}z^{w_{\scaleto{\mathrm{H}}{3pt}}\left(\ed{\vecu\G}\right)}\ed{\vecy^{\vecu\G}} \right)^{\nu-1}\\
	    &\qquad\left(\sum_{\vecu\in\{0,1\}^{k_{\scaleto{1}{3pt}}}} z^{w_{\scaleto{\mathrm{H}}{3pt}}\left(\ed{\vecu\G}\right)}\ed{\vecy^{\vecu\G}}\right) \label{eq:iowef_step3}%\\
	\end{align}
	where the result follows by applying the similar steps in the proof of Theorem \ref{theorem:wef_spc}.}
\end{proof}
\ed{Similar to the \ac{WEF}, one can compute the \ac{IOWEF} of short and moderate-length \ac{SPC} product codes iteratively, by choosing one component code to be a \ac{SPC} product code $\code_1$ and the other one to be a \ac{SPC} code $\SPC_\nu$. Given the average weight enumerator of a concatenated ensemble, upper bounds on the \ac{ML} decoding error probability can be obtained as in Section \ref{sec:wef_SPC_PC}.}

Fig. \ref{fig:scl_125_56_SPC} shows the performance of concatenating the $(125,64)$ product code with \ed{a $8$-bit} outer \ac{CRC} code with generator polynomial \ed{$g(x)=x^8+x^6+x^5+x^4+x^2+x+1$}, where the interleaver between the codes is the trivial one defined by an identity matrix. This concatenation leads to a $(125,56)$ code. Since the code distance properties are improved (observed directly via the \ac{TSB} on the average performance of the code ensemble), the performance improvement under \ac{ML} decoding is expected to be significant. \ed{The \ac{CRC} polynomial is selected as the one which provides the best \ac{TSB}, obtained for the average weight enumerator of a concatenated ensemble with the uniform interleaver assumption~\cite{BenMon96}.} The gain achieved by \ac{SCL} decoding is remarkably large, operating below the \ac{TSB}. At  a \ed{\ac{BLER}} of $10^{-2}$, \ac{SCL} decoding of the concatenated code achieves gains up to $1.25$ dB over the original product code, reaching up to \ed{$1.5$} dB at a \ed{\ac{BLER}} $\approx 10^{-4}$. The gap to the \ac{RCU} bound is approximately \ed{$0.7$} dB at \ed{a \ac{BLER}} of \ed{$10^{-4}$ and less for higher \acp{BLER}, providing a competitive performance for similar parameters\cite{Coskun18:Survey}. For example, the performance of a $(128,64)$ \own{5G-NR} LDPC code (base graph 2, see \cite{Coskun18:Survey}) is reported. The code has been decoded with the BP algorithm by setting the maximum number of iterations to $100$. The concatenation of the outer CRC code with the inner SPC product code yields a remarkable gain of $\approx 0.6$ dB with respect to the \own{5G-NR} LDPC code (it shall be noted, however, that the concatenated SPC code possesses a slightly lower code rate)}. Note that it is not always possible to attain \ed{a} performance close to the \ac{ML} performance of concatenated codes using \ac{BP} decoding\cite{Coskun19:RM-Product,geiselhart2020crcaided}. In that sense, \ac{SCL} decoding provides a low-complexity solution to \ed{approach} the \ac{ML} performance of the concatenated \ed{\ac{SPC} product code} scheme. \ed{As another reference, a (125, 56) CRC-concatenated polar code is constructed by using the $(125,64)$ polar code with an $8$-bit outer \ac{CRC} code optimized using the guidelines of \cite{Yuan19}, which takes into account the exact code concatenation under \ac{SCL} decoding rather than an ensemble performance. The generator polynomial of the CRC code is $g(x)=x^8+x^7+x^6+x^5+1$. The gap between the \ac{ML} performance of two codes is less than $0.5$ dB in the considered regime. A careful optimization of the interleaver of the CRC-concatenated \ac{SPC} product code might provide further gains as for polar codes\cite{RBC18}, but it is not in the scope of this work.}
\begin{figure}[t]
	\begin{center}
		\includegraphics[width=\columnwidth]{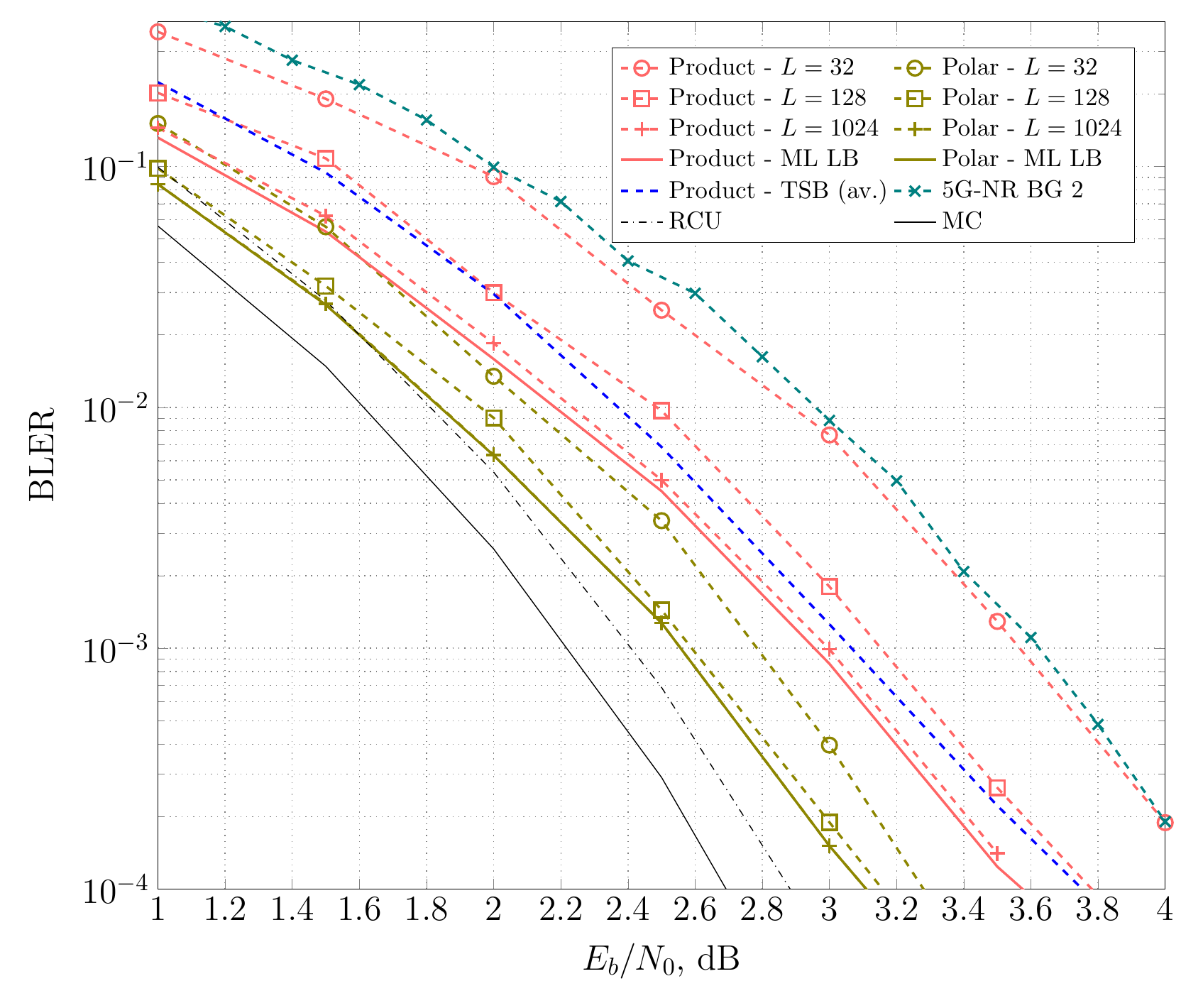}
	\end{center}
	\vspace{-0.5cm}
	\caption{\ed{\ac{BLER} vs. \ac{SNR} over the \ac{B-AWGN} channel for a $(125,64)$ product code concatenated with a $(64,56)$ \ac{CRC} code under \ac{SCL} decoding with various list sizes. The performance is compared to the one of a $(125,56)$ CRC-concatenated polar code with various list sizes, where the generator polynomial of the outer code is optimized for \ac{SCL} decoding\cite{Yuan19}, and the one of the $(128,64)$ \own{5G-NR} LDPC code (base graph 2, see \cite{Coskun18:Survey}) under BP decoding where the maximum number of iterations is set to $100$.}
	}\label{fig:scl_125_56_SPC}
\end{figure}

%% file: Conclusions.tex
\section{Conclusions}
\label{sec:conc}

Successive cancellation (SC) decoding was introduced for the class of product codes obtained by iterating single parity-check (SPC) codes, namely SPC product codes. As a byproduct, \ed{the relations between \ac{SPC} product and multi-kernel polar codes were studied, which enables us to use tools of the latter.} SPC product codes have been, then, analyzed under \ac{SC} \ed{decoding} over the binary \ed{memoryless symmetric channels, in particular, over the binary-input additive white Gaussian (B-AWGN) channel and the binary erasure channel (BEC)}.

Successive cancellation list (SCL) decoding has been described for \ac{SPC} product codes as well, which outperforms belief propagation (BP) decoding even with a considerably small list sizes, e.g., $L=4$, for the demonstrated example \ed{over the B-AWGN channel}. Larger gains are attained by concatenating an inner \ac{SPC} product code with a high-rate outer code. The \ed{maximum-likelihood} performance \ed{(with an outer code)} was analyzed via \ed{(average)} distance spectrum by developing an efficient method to compute \ed{the (input-output) weight enumerator} of short \ac{SPC} product codes, \ed{which is shown to be achieved via \ac{SCL} decoding with moderate list sizes for the considered example.}

\ed{While the asymptotic analysis over the BEC under SC decoding has shown, for some specific product code sequences, a considerable gap to channel capacity, finite-length results on \ac{SPC} product codes show that they are well-suited for SC and SCL decoding, with performance gains over BP decoding that are especially visible when an outer error detection code is used in combination with the list decoder.}

%% file: appendix.tex
\appendix
\subsection{Proof of Lemma 2}\label{sec:appendix}

   \ed{Let $U_1^{n_\ell}$ be a-priori uniform on $\mathcal{X}^{n_\ell}$ that is mapped onto $X_1^{n_\ell}$ as $X_1^{n_\ell} = U_1^{n_\ell}\GK_{n_\ell}$. It is sufficient to show that\cite[Theorem 1]{benammar_land} there exists $\alpha,\beta>0$ for all $i\in[n_\ell]$ such that
    \begin{equation}
        \left|I\left(\W^{(i)}_{\GK_{n_\ell}}\right)-I\left(\W\right)\right|\geq I\left(\W\right)^\alpha\left(1-I\left(\W\right)\right)^\beta
    \end{equation}
    which can be equivalently translated into
    \begin{align}
        \!\!\!\!\!\!\!\!\left|H\left(U_i\big|Y_1^{n_\ell},U_1^{i-1}\right)-H(\W)\right|\geq\left(1-H(\W)\right)^\alpha H^\beta(\W).
        \label{eq:second_condition}
    \end{align}
    }
    
    \ed{Note that $H\left(U_i\big|Y_1^{n_\ell}, U_1^{i-1}\right)$ is decreasing in $i$ for $i\geq 2$ due to the kernel structure. This means, for $i\in\{2,\ldots,n_\ell\}$, we have
    \begin{equation}
        H\left(U_i\big|Y_1^{n_\ell},U_1^{i-1}\right) \leq H\left(U_2\big|Y_1^{n_\ell},U_1\right).\label{eq:monotone}
    \end{equation}
    In the following, we focus on the \ac{RHS} of \eqref{eq:monotone}. By noting $U_1 = X_1\oplus\ldots\oplus X_{n_\ell}$ and $U_2 = X_2$ and setting $S = U_1\oplus X_2 = X_1\oplus X_3\ldots\oplus X_{n_\ell}$, we re-write the \ac{RHS} of \eqref{eq:monotone} as
    \begin{align}
        H&\left(X_2\big|Y_1^{n_\ell},S\oplus X_{2}\right) = H(X_2|Y_1^{n_\ell})\!+\!H(S\oplus X_{2}|Y_1^{n_\ell},X_2)\\
        &\qquad\qquad\quad\qquad\qquad\qquad\qquad-H(S\oplus X_{2}|Y_1^{n_\ell})\label{eq:separation23}\\
        &\qquad\,= H(X_2|Y_2) + H(S|Y_{\sim 2}) - H(S\oplus X_2|Y_1^{n_\ell}) \label{eq:separation24}
    \end{align}
    where \eqref{eq:separation23} follows from the chain rule of entropy and \eqref{eq:separation24} from the fact that $X_2-Y_2-Y_{\sim 2}$ and $T-Y_{\sim 2}-Y_2$ form Markov chains with $Y_{\sim i}$ being the random vector where the $i$-th element is removed. An upper-bound on the last term in the \ac{RHS} of \eqref{eq:separation24} can be found as
    \begin{align}
        &\!\!\!\!\!\!\!\!\!\!H(S\!\oplus \!X_2|Y_1^{n_\ell}) \!\geq\! H_2\!\left(\!H_2^{-1}\!\left(H(S|Y_{\sim 2})\right) \!*\! H_2^{-1}\!\left(\!H(X_2|Y_2)\right)\right)\label{eq:separation31}\\
        &\qquad\qquad\geq H(S|Y_{\sim 2}) + [1-H(S|Y_{\sim 2})]H^2\!\left(X_2|Y_2\right)\label{eq:separation32}
    \end{align}
    where \own{$H_2: [0,\nicefrac{1}{2}] \rightarrow [0,1]$ is the binary entropy function. Then,} \eqref{eq:separation31} and \eqref{eq:separation32} are due to Mrs. Gerber's Lemma \cite{WZ73} and $H_2(a*b) \geq H_2(a) + [1-H_2(a)]H_2^2(b)$,\footnote{\ed{Note that this inequality is very similar to the one given  in \cite[Eq. (38)]{benammar_land} without proof. Its validity can be easily verified numerically.}} respectively. \own{W}e subtract $H(\W)$ from the both sides of \eqref{eq:separation24} to have
    \begin{align}
        \!\!H&\!\left(\!X_2\big|Y_1^{n_\ell},S\oplus X_{2}\right) \!-\! H(\W)\!\leq\! H(S|Y_{\sim 2}) \!-\! H(S\!\oplus\! X_2|Y_1^{n_\ell})\label{eq:separation41}\\
        &\qquad\,\,\leq -[1-H(S|Y_{\sim 2})]H^2(\W)\label{eq:separation42}\\
        &\qquad\,\,= -[1-H(X_1\oplus X_3\oplus\ldots\oplus X_{n_\ell}|Y_{\sim 2})]H^2(\W)\label{eq:separation43}
    \end{align}
    where \eqref{eq:separation42} follows from \eqref{eq:separation32}.%, and \eqref{eq:separation43} by recalling $ S =  X_1\oplus X_3\oplus\ldots\oplus X_{n_\ell}$.
    For an upper-bound on \eqref{eq:separation43}, we write
    \begin{align}
        &H(\own{S}|Y_{\sim 2}) = \sum_{\vecy_{\sim 2}}p(\vecy_{\sim 2})H(\own{S}|Y_{\sim 2} = \vecy_{\sim 2})\label{eq:separation51} \\
        &\!= \!\sum_{\own{\vecy_{\sim 2}}} p(y_1)p(y_3^{n_\ell})H_2\left(p_{X_1\oplus X_3\oplus\ldots\oplus X_{n_\ell}|Y_{\sim 2} = \vecy_{\sim 2}}\right)\label{eq:separation52} \\
        &\!= \!\sum_{\own{\vecy_{\sim 2}}} p(y_1)p(y_3^{n_\ell})H_2\!\left(p_{X_1|Y_1=y_1}\!*\!p_{X_3\oplus\ldots\oplus X_{n_\ell}|Y_3^{n_\ell} = y_3^{n_\ell}}\!\right)\label{eq:separation53} \\
        &\leq H(X_1|Y_1) \!+\! [1\!-\!H(X_1|Y_1)]H(X_3\!\oplus\!\ldots\!\oplus\! X_{n_\ell}|Y_3^{n_\ell})\label{eq:separation54} \\
        &= H(\W) + [1-H(\W)]H(X_3\oplus\ldots\oplus X_{n_\ell}|Y_3^{n_\ell})\label{eq:separation55}
    \end{align}
    where \eqref{eq:separation52} \own{follows by recalling $ S =  X_1\oplus X_3\oplus\ldots\oplus X_{n_\ell}$ and defining $p_{S|Y_{\sim 2} = \vecy_{\sim 2}}\triangleq H_2^{-1}(H(S|Y_{\sim 2} = \vecy_{\sim 2}))$,} \eqref{eq:separation53} follows \own{from} the independence of $\{Y_i\}$, $i\in[n_\ell]$, and \eqref{eq:separation54} from the inequality  $H_2(a*b)\leq H_2(a)+[1-H_2(a)]H_2(b)$.\footnote{\ed{This inequality can also be found in\cite[Eq. (45)]{benammar_land} and verified numerically.}} By recalling $X_2 = U_2$ and $U_1 = S\oplus X_2$, combining \eqref{eq:separation55} and \eqref{eq:separation43} results in
    \begin{align}
        &H\left(U_2\big|Y_1^{n_\ell},U_1\right) - H(\W) \leq\\
        &\!-\!\!\left[\!1\!-\!\left(\!H(\W) \!+\! [1\!-\!H(\W)]H(X_3\!\oplus\!\ldots\!\oplus\! X_{n_\ell}|Y_3^{n_\ell})\right)\!\right]\!H^2(\W) \\
        &\!= -\!\left(1-H(X_3\oplus\ldots\oplus X_{n_\ell}|Y_3^{n_\ell})\right)[1-H(\W)]H^2(\W)\label{eq:separation56}
    \end{align}
    where \eqref{eq:separation56} follows from algebraic manipulation. Now recall \eqref{eq:monotone}. By recursively applying the steps to reach \eqref{eq:separation56}, one obtains
    \begin{equation}
        H\left(U_i\big|Y_1^{n_\ell},U_1^{i-1}\right)-H(\W) \leq -[1-H(\W)]^{n_\ell - 1} H^2(\W)\label{eq:separation61}
    \end{equation}
    for $i = 2,\ldots,n_\ell$.}
    
    \ed{Using the chain rule for conditional entropy, we write
    \begin{align}
        \sum_{i=1}^{n_\ell} \left[H(U_i|Y_1^{n_\ell},U_1^{i-1})\!-\!H(\W)\right] &= H(U_1^{n_\ell}|Y_1^{n_\ell}) \!-\! n_\ell H(\W)\\
        &= 0. \label{eq:separation71}
    \end{align}
    Combining \eqref{eq:separation61} and \eqref{eq:separation71} provides
    \begin{align}
        H\left(U_1\big|Y_1^{n_\ell}\right)-H(\W) &\geq (n_\ell-1)[1-H(\W)]^{n_\ell - 1} H^2(\W) \\
        &\geq [1-H(\W)]^{n_\ell - 1} H^2(\W) \label{eq:separation81}
    \end{align}
    where \eqref{eq:separation81} follows from the fact that $n_\ell\geq 2$. We obtain \eqref{eq:second_condition} by setting $\alpha = n_\ell - 1$ and $\beta = 2$, which concludes the proof.}